%% file: delayedoutcomes.tex
\documentclass[hidelinks, 11 pt]{article}
\usepackage{amsfonts,amsmath,amssymb,amsthm}
\usepackage{appendix}
\usepackage{subfigure}
\usepackage[pdftex]{graphicx}
\usepackage{natbib}
\usepackage[hypertexnames=false]{hyperref}
\usepackage{setspace}
\usepackage{pdfsync}
\usepackage{lscape}
\usepackage{verbatim}
\usepackage{enumerate}
\usepackage[hmargin=1in,vmargin=1in]{geometry}

\usepackage{booktabs, multirow, rotating, threeparttable} 

\hypersetup{colorlinks, linkcolor = {teal}, citecolor = {blue}, urlcolor ={blue!80!black}}

\makeatletter
\renewcommand\paragraph{\@startsection{paragraph}{4}{\z@}%
            {-2.5ex\@plus -1ex \@minus -.25ex}%
            {1.25ex \@plus .25ex}%
            {\normalfont\normalsize\bfseries}}
\makeatother
\usepackage{mathrsfs,dsfont}
\usepackage[dvipsnames,svgnames, x11names]{xcolor}
\usepackage{tikz,graphicx,pgfplots,pgfkeys} 
\usetikzlibrary{patterns}
\usepackage{caption}

\newenvironment{customlegend}[1][]{
        \begingroup
      \csname pgfplots@init@cleared@structures\endcsname
        \pgfplotsset{#1} }{
                \csname pgfplots@createlegend\endcsname
        \endgroup
    } 
\def\addlegendimage{\csname pgfplots@addlegendimage\endcsname}

\DeclareGraphicsExtensions{.jpg}
\DeclareUnicodeCharacter{202F}{\,}

\def\qed{\rule{2mm}{2mm}}

\theoremstyle{definition}

\newtheorem{theorem}{Theorem}[section]
\newtheorem{lemma}{Lemma}[section]
\newtheorem{example}{Example}[section]
\newtheorem{definition}{Definition}[section]

\newtheorem{assumption}{Assumption}[section]

\newtheorem{remark}{Remark}[section]

\let\oldmarginpar\marginpar
\renewcommand{\marginpar}[2][rectangle,draw,fill=black, text=white,text width= 2cm,rounded corners]{
    \oldmarginpar{
    \tiny \tikz \node at (0,0) [#1]{#2};}
    }
\pgfplotsset{compat=1.17}

\DeclareMathOperator*{\cov}{Cov}

\DeclareMathOperator*{\var}{Var}

\begin{document}
\author{
Federico A. Bugni \\
Department of Economics\\
Northwestern University \\
\url{federico.bugni@northwestern.edu}
\and
Ivan A.\ Canay \\
Department of Economics\\
Northwestern University\\
\url{iacanay@northwestern.edu}
\and
Steve McBride \\
Head of Economy Science\\ Roblox
}

\bigskip

\title{Decomposition and Interpretation of Treatment Effects in Settings with Delayed Outcomes\footnote{We thank Chuck Manski, Peng Ding, Lihua Lei, Filip Obradovic, and participants at various seminars and conferences for comments. We also thank Danil Fedchenko for excellent research assistance.}}

\maketitle

\vspace{-0.4in}

\begin{spacing}{1.1}
\begin{abstract}
This paper studies settings where the analyst is interested in identifying and estimating the average \emph{direct} causal effect of a binary treatment on an outcome. We consider a setup in which the outcome realization does not get immediately realized after the treatment assignment, a feature that is ubiquitous in empirical settings. The period between the treatment and the realization of the outcome allows other observed actions to occur and affect the outcome. In this context, we study several regression-based estimands routinely used in empirical work to capture the average treatment effect and shed light on interpreting them in terms of ceteris paribus effects, indirect causal effects, and selection terms. We obtain three main and related takeaways under a common set of assumptions. First, the three most popular estimands do not generally satisfy what we call \emph{strong sign preservation}, in the sense that these estimands may be negative even when the treatment positively affects the outcome conditional on any possible combination of other actions. Second, the most popular regression that includes the other actions as controls satisfies strong sign preservation \emph{if and only if} these actions are mutually exclusive binary variables. Finally, we show that a linear regression that fully stratifies the other actions leads to estimands that satisfy strong sign preservation.

\end{abstract}
\end{spacing}



\thispagestyle{empty} 
\newpage
\setcounter{page}{1}

\section{Introduction}

We study settings where the analyst is interested in identifying and estimating an average \emph{direct} causal effect of a binary treatment on an outcome, and the treatment status is determined in the context of a randomized controlled experiment or an observational study under conditional independence assumptions. We focus on settings where the outcome of interest does not get immediately realized after treatment assignment, a feature that is ubiquitous in empirical settings, including the study of long-run effects in economic history \citep{voigtlander/voth:12,angelucci/etal:22}, the analysis of long-term valuation in industry settings \citep{jain2002customer,chen:2021}, and randomized experiments in economics and other social sciences \citep{Beaman/etal:2013,Moderna}. The delay in the realization of the outcomes creates a time window between the treatment assignment and the realization of the outcome that, in turn, opens up the possibility for other observed ex-post actions to take place before the outcome is finally realized; see Figure \ref{fig:timeline} for a graphical representation. In this context, we study the interpretation of several popular estimands that arise from running regressions of the outcome on the treatment and different ways of ``controlling'' for the other actions. We emphasize that our use of the term ``regression'' refers to a linear projection that is free of any modeling assumptions on potential outcomes or conditional means. Some of these estimands are not only popular in the economics literature, see, e.g., \citet{fagereng2021wealthy,heckman2013understanding,chernozhukov/etal:21joe}, but are also widely used across other sciences, like psychology, epidemiology, biostatistics, and political science, as shown by the large number of citations associated with the regression approach popularized by \cite{baron1986moderator}. For each of these estimands, our results present a decomposition that facilitates their interpretation in terms of ceteris paribus effects of the treatment on the outcomes, indirect effects caused by the other actions, and selection terms; and provide a framework that allows us to clarify under what type of conditions the practice of ``controlling'' for the presence of other actions leads to estimands that admit the desired interpretation. 

The main findings of this paper can be grouped into three sets of results. First, the standard practice of studying estimands that arise from a regression of an outcome on the treatment, with or without ``controlling'' for the other actions in such regressions, does not generally satisfy what we call \emph{strong sign preservation}. Strong sign preservation, formally characterized in Definition \ref{def:strong-SP}, is satisfied when an estimand that intends to measure a \emph{ceteris paribus} (or direct) causal effect of a treatment on an outcome is positive when the effect of the treatment on the outcome is positive conditional on \emph{all} possible values of the other actions. Failure to satisfy strong sign preservation introduces a Simpson’s paradox-like sign reversal where the estimands may be negative even when the treatment positively affects the outcome for any possible combination of other actions. Second, the most popular estimand that linearly controls for the other actions in the regression, and that we label the long regression, does not generally provide benefits relative to the short regression that includes no controls whatsoever. More concretely, while neither the short nor the long regression satisfies strong sign preservation, the estimand associated with the long regression admits a decomposition in terms of weighted averages of well-defined causal effects but where the weights could potentially be negative. This feature introduces yet another source that may separate the sign of the estimand from the sign of ceteris paribus causal effects. Notably, this feature also occurs when the regression includes interaction terms between the treatment and the other actions. Perhaps our most salient result is the one in Theorem \ref{thm:long}, which shows that the long regression delivers easy-to-interpret results \emph{if and only if} the other actions are all binary and mutually exclusive random variables. Finally, while non-parametric identification of the effects is straightforward under the stronger form of our assumptions and follows directly from a saturated regression, we also show that a linear regression that properly controls for other actions through complete stratification produces estimands that satisfy strong sign preservation. We term this the “strata fixed effects regression” due to its link to the practice of including strata fixed effects in randomized controlled trials with covariate adaptive randomization (see \cite{bugni/canay/shaikh:2018,bugni/canay/shaikh:2019}).

The decompositions we derive for each of the five regression-based estimands can be interpreted as a partition into ``direct'' and ``indirect'' effects (and possibly ``selection'' effects depending on the assumptions), and so our results are linked to the vast literature on mediation analysis, see, e.g., \cite{baron1986moderator}, \cite{pearl:01}, \cite{robins:03}, \cite{kaufman2004critique}, \cite{kaufman2009gilding}, \cite{imai2010}, \cite{vanderweele/etal:14}, and Remark \ref{rem:mediator} for a discussion. See \cite{vanderweele-book:15} for a book-level treatment. However, as opposed to the literature on mediation that studies the type of assumptions that would identify the causal effects of the so-called \emph{mediators}, which in our context would simply be the other actions taken before the outcome is realized, here our goal is not to identify these indirect effects but rather to gain a better understanding of how to properly interpret certain popular estimands of the effect of the treatment on the outcome without the aid of functional form assumptions. 

Beyond the literature on mediation analysis, our paper also connects to several strands of research in econometrics and biostatistics. First, we are not the first to highlight the importance of distinguishing between “direct” (or ``partial'') and “total” causal effects in econometrics. Early discussions include those in \citet{Manski:97} and \citet{heckman:00}; see also Remark \ref{rem:ssp-comment}. Second, our results relate to a large body of research that interprets regression-based estimands in various settings, often characterizing them as weighted averages of causal effects of interest—sometimes involving negative weights. Examples include the literature on two-way fixed effects in difference-in-differences settings \citep{DeChaisemartin/xavier:20,goodman-bacon:21,callaway/santanna:21,sun/abraham:21,borusyak/jaravel:22}; on local instrumental variables and marginal treatment effects \citep{canay/etal:2023}; and on contamination bias when conditioning on covariates \citep{angrist:98,goldsmithpinkham/etal:22}. Our paper is also connected to the literature on design-based analyses of regression estimands in factorial experiments \citep{zhao/ding:22}. These connections are primarily algebraic—stemming from shared properties of least squares—though the questions we address and the main takeaways of our analysis are distinct. Third, our results on the failure of strong sign preservation for popular estimands resonate with the “surrogate paradox” in the surrogacy literature, where a treatment may have a positive effect on a surrogate, and the surrogate may be positively correlated with the outcome, yet the treatment effect on the outcome may be negative; see \citet{chen/etal:07} and \citet{vanderweele:13}. Finally, the decomposition we derive for the short regression mirrors those developed in the literature on causal interactive effects—where units may be subject to multiple types of exposures (e.g., genetic and environmental)—as in \citet{robins1992identifiability} and \citet{vanderweele:14}. Our focus on a broader range of regression estimands, however, sets our work apart.

The remainder of the paper is organized as follows. Section \ref{sec:notation} introduces the basic notation. Section \ref{sec:causal-effects} defines the main concepts we use throughout the paper, including partial causal effects, direct causal effects, and strong sign preservation. Section \ref{sec:decomposition} introduces the five estimands we study  and then presents the main results on how each of these estimands admits different decompositions into direct, indirect, and selection effects. Finally, Section \ref{sec:conclusion} concludes.  

\section{Setup and Notation}\label{sec:notation}

Consider a setting where  $Y$  denotes the observed outcome of interest, and the actions taken by individuals or units under study are divided into a ``main'' action of interest, denoted by  $D$, and ``other'' actions, denoted by $A$:
\begin{equation}\label{eq:DA}
    (D,A) \in \mathcal D \times \mathcal  A ~.
\end{equation}
Let  $X$  represent other observed covariates, which include features beyond actions.

All actions are assumed to be discrete. The main action,  $D$, is further assumed to be binary, i.e.,  $\mathcal D \equiv \{0,1\}$. The other actions,  $A$, form a  $K$-dimensional vector taking values in \( \mathcal A \equiv \{a=(a_1,\dots,a_{K}): a_j\in \mathcal A_j \text{ for } j=1,\dots,K \} \), where each $ \mathcal A_j$  is a finite set. For notational convenience, we assume $ \mathcal A_j \subseteq \mathbb N$,  $0 \in \mathcal A_j$, and that when  $|\mathcal A_j| = 2$, we have $ \mathcal A_j = \{0,1\}$. These additional restrictions simplify the discussion and expressions but are not required for our results.

The setting we study in this paper is one with the following characteristics. First, the analyst controls the action of interest $D$ via either a randomized controlled experiment or, alternatively, by an exogeneity assumption like selection on observables. We therefore alternatively call this action the ``treatment''. Second, the outcome $Y$ is not instantaneous and takes some time to be realized within the timeline of the experiment. In the period between treatment assignment and outcome realization, the other actions in $A$ are chosen by the units—these may be modeled as conditionally exogenous in some applications, or left unmodeled altogether. Figure \ref{fig:timeline} illustrates the setting. Below we describe some empirical applications in economics, social sciences, and industry that naturally fit into this setting. 

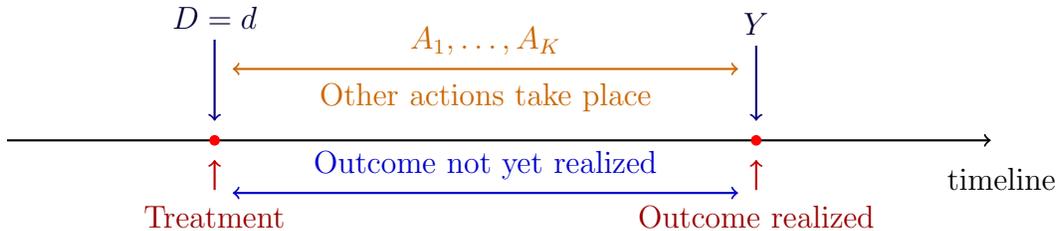
\begin{figure}
\large
~~~~~~~~~~~\begin{tikzpicture}[]
 
  \newcount\nOne; \nOne=-6
  \def\w{12}      
  \def\n{5}       
  \def\noffset{1} 
  \def\nskip{0}   
  \def\la{2.00}   
  \def\lt{0.20}   
  \def\ls{0.15}   
 
  \def\myx(#1){{(#1-\nOne)*\w/\n}}
  \def\arrowLabel(#1,#2,#3,#4){
    \def\xy{(#1-\nOne)*\w/\n}; \pgfmathparse{int(#2*100)};
    \ifnum \pgfmathresult<0
      \def\yyp{{(\lt*(-0.10+#2))}}; \def\yyw{{(\yyp-\la*\lt*#3)}}
      \draw[<-,thick,black!50!blue,align=center]
        (\myx(#1),\yyp) -- (\myx(#1),\yyw)
        node[below,black!80!blue] {#4}; 
    \else
      \def\yyp{{(\lt*(0.10+#2)}}; \def\yyw{{(\yyp+\la*\lt*#3)}}
      \draw[<-,thick,black!50!blue,align=center]
        (\myx(#1),\yyp) -- (\myx(#1),\yyw)
        node[above,black!80!blue] {#4};
    \fi}
  \def\arrowLabelRed(#1,#2,#3,#4){
    \def\yyp{{(\lt*(-0.10+#2))}}; \def\yyw{{(\yyp-\la*\lt*#3)}}
    \fill[red,radius=2pt] (\myx(#1),0) circle;
    \draw[<-,thick,black!25!red,align=center]
      (\myx(#1),\yyp) -- (\myx(#1),\yyw)
      node[below,black!40!red] {\strut#4}; 
    }
 
  \draw[->,thick] (-\w*0.03,0) -- (\w*1.06,0)
                  node[right=4pt,below=6pt] {timeline};
 
  \arrowLabel( -5,1.2,2.7,$D=d$) 
  \arrowLabel(-2.00,1.2,2.5,$Y$)     
 
  \arrowLabelRed(-5,-1.2,1.0,Treatment) 
  \arrowLabelRed(-2,-1.2,1.0,Outcome realized) 
 
  \draw[<->,thick,black!20!orange]
    ({(-4.9-\nOne)*\w/\n},0.95) -- ({(-2.1-\nOne)*\w/\n},0.95)
    node[midway,below=1pt] {Other actions take place}
    node[midway,above=1pt] {$A_1,\dots,A_K$};
 \draw[<->,thick,black!20!blue]
    ({(-4.9-\nOne)*\w/\n},-0.7) -- ({(-2.1-\nOne)*\w/\n},-0.7)
    node[midway,above=1pt] {Outcome not yet realized};
 
\end{tikzpicture}
\caption{{\small Timeline of actions. The first action, $D$, is assumed to be (conditionally) exogenous and is the main action of interest (so we refer to it as the treatment). The outcome is not instantaneous and may take a short or long period of time to get realized. In the meantime, units choose the value of the other actions $A_1,\dots,A_K$.}}\label{fig:timeline}
\end{figure}

The first class of applications that fit our framework is the literature that studies long-run outcomes in economics. There, interest typically lies in a treatment that happened several years in the past (oftentimes hundreds of years ago) on some outcome of interest in more recent times. For example, \cite{voigtlander/voth:12} study the effect of the existence of Black Death pogroms in 1349 ($D$) on the level of anti-Semitism in Nazi Germany in 1920s ($Y$) at the city level. Other actions that are included in the analysis include city population in the 1920s ($A$), among other variables that are determined closer to the realization of the outcomes. Other examples include \citet{nunn:08} and \cite{angelucci/etal:22}, among many others. In these applications, it is common to argue that the treatment is exogenous invoking selection on observables assumptions or, in some cases, relying on instrumental variables approaches. Our results are limited to settings exploiting conditional exogeneity. 

The second class of applications arises in the analysis of long-term valuation in marketing and industry. In these settings, companies such as Airbnb, Uber, Google, and Microsoft operate platforms where users engage in a wide range of actions—such as making a booking, writing a review, subscribing to a service, or watching a video—and are interested in estimating the long-run value that a specific product or action generates for the firm. For instance, \citet{chen:2021} discusses how Airbnb measures the long-term value of user actions by tracking the revenue generated over a 365-day window. The example focuses on the effect of a guest making a booking ($D$) on long-term revenue ($Y$), while other intermediate actions—such as cancellations or writing reviews—are captured in the vector $A$. Although randomized experiments are often feasible for some actions, others cannot easily be randomized due to ethical, legal, or user experience concerns. In such cases, it is common to rely on conditional exogeneity (i.e., selection on observables) and focus on one action of interest at a time. Importantly, in these applications, the goal is often to identify a \emph{direct} causal effect of the focal action ($D$) on the outcome ($Y$), holding other actions ($A$) fixed. This approach helps avoid the problem known as \emph{double counting}—that is, attributing the same impact multiple times through correlated or causally related behaviors. For this reason, isolating the direct effect of $D$ is not merely of academic interest, but a practical necessity in how firms estimate the marginal contribution of individual features or interventions. Related examples and discussions can be found in \citet{Lewis/etal:11} and \cite{xu/etal:2015}, who discuss how correlated consumer behaviors—termed ``activity bias”—can lead to substantial overestimation of advertising effectiveness.

The third class of applications includes clinical trials and randomized experiments with a follow-up period between the treatment assignment and the realization of the outcome of interest. There, researchers typically fully control the assignment of the main treatment of interest but are unable to restrict behavior in the follow-up period. For example, the clinical trial run by \cite{Moderna} to study the efficacy of the Moderna COVID-19 vaccine against SARS-CoV-2 infections. Participants in the study were randomized to Immediate Vaccination Group 1 (receiving the Moderna COVID-19 Vaccine on Day 1 and Day 29) or Standard of Care Group 2, with vaccination given at months 4 and 5. During the months following vaccinations, participants received visits that checked for infections and could include blood collection, nasal swabs, SARS-CoV-2 screening, COVID-19 symptom checks, and questionnaires. In this example, $D$ would be an indicator of whether the participant received a vaccine, $Y$ would be an indicator of whether the participant got infected within the 4 months of the study, and $A$ would include other actions taken by the participants that could affect infection rates, like whether the participants wear masks in public, whether the participants avoid large gatherings, etc.\footnote{In this paper we abstract away from spillover effects between individuals, which could be relevant in the context of this and other examples we describe.} Another example is \cite{Beaman/etal:2013}, who conducted a randomized field experiment that provides free fertilizer to female rice farmers in Mali to measure how they choose to use the fertilizer and the overall impact on profitability. The authors distributed the fertilizer in May 2020 and conducted two follow-up surveys, one in August 2020 and one in December 2020, right after the harvest. In this example, $D$ would be an indicator of whether the farmer received free fertilizer, $Y$ would be a measure of output like crop yield or just profits, and $A$ would include all relevant complementary agricultural inputs, such as labor, herbicides, and water usage. Due to its simplicity, we use this farming example to illustrate the concepts we define in the next section.

\begin{remark}\label{rem:mediator}
    What we call the other actions in Figure \ref{fig:timeline} can be alternatively labeled as ``mediator'' variables since these are post-treatment variables that occur before the outcome is realized, see, e.g., \cite{baron1986moderator}, \cite{pearl:01}, \cite{robins:03}, \cite{imai2010},  \cite{vanderweele/etal:14}, and \cite{vanderweele-book:15}, among many others. However, our work deviates from this literature in two important ways. First, while the literature on causal mediation analysis focuses on the identification of causal effects induced by mediators, our focus in this paper is to understand whether common estimands that are used to capture causal effects of main action $D$ on the outcome $Y$ admit clear interpretations through the lens of total and direct effects. Second, our decompositions in terms of direct and indirect effects are defined in terms of potential outcomes for all of the actions, including those that may be labeled as mediators, and this implies ``indirect'' effects in our context do not coincide with the definition of indirect effects in the mediation literature but rather with the so-called ``controlled'' effects discussed by \cite{pearl:01} and \cite{robins:03}; see Remark \ref{rem:controlled-effects} for additional discussion on this distinction. It is worth noting, however, that several of our results have implications for the causal mediation literature, and we discuss these implications as we present our main results.  \qed
\end{remark}

\begin{remark}\label{rem:beyond-mediator}
   The problems associated with the presence of the other ``endogeneous'' actions $A$ could arise in settings that do not require the presence of ``delayed outcomes''. Throughout the paper, we use the term ``endogenous’’ to refer to actions that may respond to the treatment, in the sense that their distribution may vary with the value of $D$. For example, the factorial experiments considered by \cite{zhao/ding:22}, the problems associated with ``bad controls'' discussed by \cite{angrist2008mostly}, or the mediation framework described in Remark \ref{rem:mediator}, are all about endogenous actions that are not necessarily related to delayed outcomes. While we do not need to invoke delayed outcomes to introduce the type of interpretation challenges we discuss here, we choose to do so because it directly speaks to the examples that motivated this paper. \qed
\end{remark}

We denote potential outcomes by $Y(d,a)$ and their expectation by 
\begin{equation}\label{eq:mu}
    \mu(d,a)\equiv E[Y(d,a)]~.
\end{equation}
We also introduce the concept of a pooled potential outcome to isolate the counterfactual outcome associated with the main action of interest (the treatment), 
\begin{equation}\label{eq:pooled} 
    Y(d) = Y(d,A(d)) = \sum_{a\in\mathcal A} Y(d,a)I\{A(d)=a\}~, 
\end{equation}
where $A(d)$ denotes potential outcomes for the actions $A$ as a function of the treatment $d$. Finally, the observed outcome $Y$ is related to potential outcomes by the relationship 
\begin{equation} \label{eq:obsy}
    Y = \sum_{(d,a)\in\mathcal D\times \mathcal A} Y(d,a)I\{(D,A)=(d,a)\}~.
\end{equation}

In the next section, we define the type of causal effects we are interested in this paper, as well as the assumptions that we invoke to interpret the decompositions we derive for a variety of estimands. 

\section{Causal Treatment Effects}\label{sec:causal-effects}
We start by discussing the type of counterfactual treatment effects that could interest the researcher in the canonical setting where $D$ is binary. Viewing $Y(d,a)$ as a function of two types of actions immediately suggests that there could be total effects, direct effects, and indirect effects, all of which may or may not be of interest in the context of a concrete application. Understanding the variety of causal effects that one could describe, in turn, will help us provide representations and interpretations of commonly used target parameters, like the average treatment effect (ATE), in terms of these types of causal effects. We start with what is perhaps one of the most natural types of \emph{ceteris paribus} effects in Definition \ref{def:PCE}.  

\begin{definition}[Average Partial Causal Effect]\label{def:PCE} An {\bf average partial causal effect} of $D$ on the outcome $Y$ is any difference of the form $\mu(1,a) - \mu(0,a)$, where the value $a\in\mathcal A$ is kept constant.
\end{definition}

Definition \ref{def:PCE} defines an average partial causal effect of the main action as a mean comparison that keeps the value of the other actions unchanged in both states of the comparison (\emph{ceteris paribus}). In the farming application of \cite{Beaman/etal:2013}, it would capture the average causal effect of using fertilizer on the crop yield, while keeping other inputs, like labor, herbicides, and water usage, constant in the counterfactual comparison. 

The \emph{ceteris paribus} effect in Definition \ref{def:PCE} could also be defined conditional on certain events or subpopulations. To account for this, we also consider the concept in Definition \ref{def:PCE} conditional on some set $\Omega$, i.e., $E[Y(1,a)-Y(0,a) \mid \Omega ]$, where $\Omega$ is a function of $(D,A,X)$. For example, $\Omega = I\{D=1\}$ would lead to an average partial causal effect on the treated and $\Omega = I\{X=x\}$ would lead to an average partial causal effect for units with covariates $x$.  

The definition of a partial causal effect for the main action delivers a potentially different causal effect for each possible value of the other actions or, alternatively, provides a collection of partial causal effects indexed by $a\in\mathcal A$. While the goal could just be to identify such a collection of effects, in many settings it may be natural to aggregate this collection of partial effects in a way that summarizes the effect of the main action on the outcome of interest. The following definition defines a direct causal effect as any weighted average of partial causal effects.   

\begin{definition}[Average Direct Causal Effect]\label{def:DCE} The {\bf average direct causal effect} of $D$ on the outcome $Y$ is any convex combination of average partial causal effects of $D$ on $Y$. That is, 
\begin{equation}\label{eq:DCE}
  \sum_{a\in \mathcal A} \omega(a) (\mu(1,a) - \mu(0,a))~,
\end{equation}
where $\omega(a) \in [0,1]$ for all $a\in \mathcal A$ and $ \sum_{a\in \mathcal A} \omega(a)=1$.
\end{definition}
In the context of our farming example with $A$ only capturing low and high water usage for simplicity, the parameter \eqref{eq:DCE} combines the average causal effect of using fertilizer on the crop yield for units with high water usage, say $A=1$, and units with low water usage, say $A=0$. Average direct causal effects could also be defined conditional on a set $\Omega$. The definition does not determine how the groups are weighted, but it requires that no group is assigned a negative weight. In this sense, any average direct causal effect satisfies \emph{strong sign preservation}, as defined below. 

 \begin{definition}[Strong Sign Preservation]\label{def:strong-SP} Let $\Delta$ be an estimand that intends to capture the causal effect of the treatment $D$ on the outcome $Y$. We say that $\Delta$ satisfies {\bf strong sign preservation} if $$\mu(1,a)-\mu(0,a)>0 \text{ for all $a\in \mathcal A$ implies }\Delta>0~.$$
 \end{definition}
In our farming example with $A$ only capturing low and high water usage, strong sign preservation of a parameter $\Delta$ implies that whenever fertilizers improve the expected crop yield both for units with high water usage and units with low water usage, $\Delta$ should be positive as well. As the name suggests, strong sign preservation does not allow for the possibility of what it is typically referred to as \emph{sign reversal}, understood as a situation where $\Delta<0$ when $\mu(1,a)-\mu(0,a)>0$ for all $a\in \mathcal A$.

A premise of this paper is that researchers employ regression-based estimands to obtain an average direct causal effect, as defined in Definition \ref{def:DCE}. That is, the primary interest is not in identifying \emph{indirect} effects—capturing the impact of $D$ on $Y$ through the actions $A$—nor in identifying the total effect, which allows the actions $A$ to vary across treatment groups. While such alternative effects may be of interest in certain applications, our aim is to clarify when the estimands described in the next section can be interpreted in line with Definition \ref{def:DCE}.

\begin{remark}\label{rem:ssp-comment}
    While strong sign preservation may be perceived as a key requirement for parameters that intend to identify partial causal effects, it may not be a reasonable requirement in settings where the counter-factual question of interest involves total effects. The distinctions between ``partial'' and ``total'' causal effects have appeared in the literature in a variety of contexts, even beyond the mediation analysis literature discussed in Remark \ref{rem:mediator}, where \cite{pearl:01} and \cite{robins:03} provide comprehensive treatments on these distinctions. For example, \cite{heckman:00} defines a causal effect as the partial derivative with respect to an isolated action. He acknowledges that assuming that an action can vary independently of all other actions is strong, yet he regards it as   ``...essential to the definition of a causal parameter''. \citet[][pages 1321 and 1323]{Manski:97}, in turn,  provides two interpretations of potential outcomes (one that keeps other actions fixed and another one that lets the other actions change in response to the main action) and clarifies that the interpretation of treatment effects depends on how we think about potential outcomes. Here, we do not dwell on discussions about the relative merits of partial or total effects but rather seek to understand whether commonly used estimands in empirical work admit either of these (commonly sought after) interpretations under different assumptions. \qed
\end{remark}

\begin{remark}\label{rem:controlled-effects}
    Our definitions of average causal partial effects and average direct causal effects align with the concept of ``controlled'' effects in the mediation analysis literature; see \citet{pearl:01}, \citet{robins:03}, and \citet{kaufman2009gilding}, among others. These differ from the more commonly emphasized “natural” direct effects. In our notation, the average natural direct effect is given by \( E[Y(1, A(d)) - Y(0, A(d))] \), where the action (or mediator) \( A \) is set to the level it would naturally take under the reference treatment \( D = d \). The average natural indirect effect, in turn, is given by \( E[Y(d, A(1)) - Y(d, A(0))] \) for $d\in\{0,1\}$. While the total effect, $E[Y(1)-Y(0)]$, decomposes into natural direct and indirect effects without additional assumptions, the decomposition into controlled direct and indirect effects generally introduces selection terms unless further assumptions are imposed—as we show in the next section. On the other hand, identifying natural effects typically requires stronger assumptions than those needed for identifying controlled effects; see \citet{vanderweele/etal:14}. \qed
\end{remark}

The decompositions we present in this paper are purely algebraic and, in their most general form, do not rely on any identification assumptions. However, to interpret the terms in these decompositions as causal effects—or to simplify their expressions—it is useful to introduce assumptions that both facilitate interpretation and connect our framework to assumptions commonly used by practitioners and related work in the literature.

We begin with a baseline assumption that allows for the identification of a \emph{total} effect of $D$ on $Y$, i.e., $E[Y(1)-Y(0)]$. This assumption is standard in the literature on treatment effects in settings that do not explicitly account for the presence of the other actions $A$.

\begin{assumption}\label{ass:basic-SOO}
For all $d \in \mathcal D$, the treatment $D$ is conditionally independent of the potential outcome $Y(d)$ given covariates $X$, i.e.,
\[
D \perp Y(d) \mid X~.
\]
\end{assumption}

Assumption \ref{ass:basic-SOO} can be justified by the design of a randomized controlled experiment \citep[e.g.,][]{Beaman/etal:2013} or by relying on a sufficiently rich set of covariates—as is often the case in industry applications—to make the conditional independence assumption credible \citep[as in][]{chen:2021}. Three points are worth emphasizing. First, as previously noted, this assumption is not required to derive our decompositions, but rather to aid in their causal interpretation. Second, as noted by a referee, one could replace Assumption \ref{ass:basic-SOO} with $D \perp (Y(d,a), A(d)) \mid X$ for all $(d,a) \in \mathcal{D} \times \mathcal{A}$. Although this yields equivalent interpretations, we favor Assumption \ref{ass:basic-SOO} as it more closely aligns with applications where other actions are not explicitly modeled. Third, while Assumption \ref{ass:basic-SOO} suffices to identify the average treatment effect of $D$ on $Y$, it does not generally support a causal interpretation of estimands intended to capture \emph{ceteris paribus} or ``direct'' effects of $D$ when $A$ is also present. For that purpose, we consider an alternative assumption that grants conditional independence of $Y(d,a)$ with respect to both $D$ and $A$ given $X$.

\begin{assumption}\label{ass:joint-SOO}
For all $(d,a) \in \mathcal D \times \mathcal A$, the pair $(D, A)$ is conditionally independent of the potential outcome $Y(d,a)$ given $X$, i.e.,
\[
(D, A) \perp Y(d,a) \mid X~.
\]
\end{assumption}

Assumption \ref{ass:joint-SOO} reframes the problem as one involving multiple (conditionally exogenous) treatments. Three comments are worth pointing out. First, this assumption is strong in the sense that it enables nonparametric identification of the full mean response function $\mu(d,a)$, as discussed in Section \ref{sec:sat}. In fact, it is equivalent to Assumptions (i) and (ii) in \citet{vanderweele/etal:14}, which are sufficient for identifying controlled direct effects. Second, as noted by a referee, one could instead impose the stronger condition $(D, A(1), A(0)) \perp Y(d,a) \mid X$ for all $(d,a)$, which leads to the same interpretations. We favor Assumption \ref{ass:joint-SOO} as it aligns better with our focus on controlled effects. Third, while it may hold in settings with a rich set of covariates, it may be questionable when only $D$ is experimentally assigned and $A$ reflects behavioral responses to $D$.

To illustrate, in the agricultural production example, if fertilizer ($D$) is randomly assigned but farmers adjust water usage ($A$) in response to unobserved shocks that also affect yields—such as weather or soil quality—then Assumption \ref{ass:joint-SOO} would likely fail. By contrast, in many industry applications, it is common to condition on high-dimensional covariates and assume that all actions, including $A$, are conditionally exogenous. In these settings, $D$ is often designated as the action of interest not because it is fundamentally different from the other actions, but because of how it enters a business decision—for example, it may correspond to the action assigned to a particular team or product line. As such, the full vector $(D,A)$ may be plausibly treated as jointly exogenous given a rich enough set of observed variables.

Importantly, while Assumption \ref{ass:joint-SOO} permits identification of \( \mu(d,a) \), we show in later sections that this is still not sufficient for interpreting common regression-based estimands as average direct causal effects of \( D \) on \( Y \), as defined in Definition \ref{def:DCE}. This underscores the fragility of such interpretations even under relatively strong identification assumptions.

\begin{remark}\label{rem:mean-assumptions}
The results in this paper are expressed in terms of conditional expectations and would still hold if the conditional independence assumed in Assumptions \ref{ass:basic-SOO} and \ref{ass:joint-SOO} were weakened to conditional mean independence. We choose to retain the stronger form of these assumptions to facilitate comparison with the typical assumptions used in the literature, which are usually stated in terms of conditional statistical independence.
    \qed
\end{remark}

\section{Decomposing Common Estimands}\label{sec:decomposition}
In this section, we analyze five natural and highly popular estimands intended to capture treatment effects of $D$ on $Y$. For each of these estimands, we derive a decomposition in terms of parameters that can be labeled according to Definitions \ref{def:PCE} and \ref{def:DCE} and discuss under what assumptions they can be interpreted as intended. To keep our exposition as simple as possible and be able to zoom in on the type of concerns we intend to highlight, in what follows we abstract away from issues related to improper control of the covariates $X$. In other words, we ignore the role of the covariates in the type of regressions we consider. This could be interpreted as a situation where the covariates are discrete, and the regressions are viewed as within-cell regressions with cells given by $X=x$, or more generally, where the covariates have been already accounted for by other means, like clustering or via a partially linear model, among many possibilities. In particular, we do not consider the possibility that the analyst improperly controls for these covariates by simply including a linear term in $X$ in the regressions, as this would create additional problems to those discussed here; see \cite{goldsmithpinkham/etal:22} for a detailed treatment on the consequences of not properly controlling for confounders. We also re-iterate that our use of the term ``regression'' refers to a linear projection that is free of any modeling assumptions on potential outcomes or conditional means.  

The first such estimand is the usual difference in means, which we write here as the slope coefficient $\Delta_{\rm short}$ in a regression (projection) of $Y$ on $D$ and a constant term,
\begin{equation}\label{eq:short-regression}
  \text{\bf Short regression:}~~~ Y = \beta + \Delta_{\rm short} D + U~,
\end{equation}
where $E[UD]=0$ by properties of projections and $E[U|D]=0$ follows from $D$ being binary. We call this the short regression. 

The second estimand is the slope coefficient $D$ in a linear regression of $Y$ on $D$, a constant term, and the $K$ actions $A_1,\dots,A_K$,
\begin{equation}\label{eq:long-regression}
  \text{\bf Long regression:}~~~ Y =  \Delta_{\rm long} D + \theta_0 + \sum_{j=1}^K \theta_{j}A_j + V~,
\end{equation}
where $E[VD]=E[VA_j]=0$ by properties of projections. We call this the long regression.

The third estimand is the slope coefficient $D$ in a linear regression of $Y$ on $D$, a constant term, the $K$ actions $A_{1},\dots ,A_{K}$, and their interactions with $D$,
\begin{equation}\label{eq:inter-regression}
  \text{\bf Long regression with interactions:}~~~ Y=\Delta_{\mathrm{inter}}D+\theta _{0}+\sum_{j=1}^{K}\theta _{j}A_{j}+\sum_{j=1}^{K}\lambda _{j}A_{j}D+e~,
\end{equation}
where $E[eD]=E[eA_j]=E[eA_jD]=0$ by properties of projections. We call this the long regression with interactions. Note that this is not a fully saturated regression in general, since the random variables $A_j$ are allowed to take arbitrary values in $\mathbb N \cup \{0\}$. 

The fourth estimand is the slope coefficient $D$ in a regression of $Y$ on $D$ and a set of indicator functions for all the values that $A$ takes,
\begin{equation}\label{eq:sfe-regression}
  \text{\bf Strata fixed effects (SFE) regression:}~~~ Y =  \Delta_{\rm sfe}  D + \sum_{a\in\mathcal A}  \theta(a) I\{A=a\} + \nu ~,
\end{equation}
where $E[\nu D]=E[ \nu I\{A=a\}]=0$ by properties of projections. Note that this is a regression of $Y$ on $D$ with ``strata fixed effects'', where the event $\{A=a\}$ defines a stratum for each value of $a$. As a result, we call this the strata fixed effect (SFE) regression. 

The last set of estimands are the slope coefficients $\Delta_{\rm sat}(a)$, for $a\in\mathcal A$, in a saturated regression of $Y$ on a set of indicator functions for all the values that $A$ takes and their interactions with $D$, 
\begin{align}
\text{\bf Saturated (SAT) regression:}~~~ Y =  \sum_{a\in\mathcal A} \gamma(a)I\{A=a\}+ \sum_{a\in\mathcal A} \Delta_{\rm sat}(a)I\{A=a\}  D + \epsilon~,\label{eq:sat-regression}
\end{align}
where $E[\epsilon D I\{A=a\}]=0$ by properties of projections and $E[\epsilon|D,I\{A=a\}]=0$ follows from $D$ and $I\{A=a\}$ being binary for all $a\in\mathcal A$. We call this the saturated (SAT) regression. 

\begin{remark}\label{rem:glynn}
   The use of short, long, and long with interaction regressions in the social science literature is ubiquitous. When \cite{glynn2012product} discusses the popularity of these regressions, he writes that long regressions are so pervasive within the social science and empirical mediation literature that ``examples are too numerous to cite.'' Indeed, \cite{baron1986moderator}, the paper that largely established the use of these and related regressions, has over $137,000$ citations as of 2025. \qed
\end{remark}

\subsection{Short regression}\label{sec:short}

The short regression is algebraically very simple, so we build up toward the main result introducing the main concepts and notation along the way. The other regressions, on the contrary, have more opaque derivations, and so in those cases, we first present the formal results and then discuss their interpretation. 

The slope coefficient $\Delta_{\rm short}$ in \eqref{eq:short-regression} equals $\Delta_{\rm short} = E[Y|D=1]- E[Y|D=0]$ by elementary arguments. If we define 
\begin{equation}\label{eq:pis}
\pi_{d}(a) \equiv P\{A=a|D=d\} ~,
\end{equation}
and note that 
\begin{equation*}
  E[Y|D=d] = \sum_{a\in \mathcal A} E[Y(d,a)|D=d,A=a]\pi_{d}(a) ~,
\end{equation*}
we can decompose $\Delta_{\rm short}$ into the following three terms,  
\begin{equation}\label{eq:short-decomposition}
  \Delta_{\rm short}=\Delta_{\rm dce}^{\rm s} + \Delta_{\rm ind}^{\rm s} + \Delta_{\rm sel}^{\rm s} 
\end{equation}
where 
\begin{align}
  \Delta_{\rm dce}^{\rm s} &\equiv \sum_{a\in \mathcal A} \pi_{1}(a) E[Y(1,a)-Y(0,a)| D=1,A=a] \label{eq:short-dce-1}\\
  \Delta_{\rm ind}^{\rm s} &\equiv \sum_{a\in \mathcal A}(\pi_{1}(a)-\pi_{0}(a))(E[Y(0,a)| D=0,A=a]-E[Y(0,0)| D=0,A=0])\label{eq:short-ind-1}\\
  \Delta_{\rm sel}^{\rm s} &\equiv \sum_{a\in \mathcal A} \pi_{1}(a) (E[Y(0,a)| D=1,A=a]-E[Y(0,a)| D=0,A=a])\label{eq:short-sel-1}~.
\end{align}
The three terms in the above decomposition for $\Delta_{\rm short}$ have a clear interpretation and show that there are two channels of endogeneity introduced by the fact that the other actions, $A$, take place in-between the treatment assignment and the realization of the outcome of interest. The term in \eqref{eq:short-dce-1}, $\Delta_{\rm dce}^{\rm s}$, captures an average direct causal effect on the treated, as in Definition \ref{def:DCE}. Note that this term conditions on $\Omega = I\{D=1,A=a\}$ and so it is a conditional effect like those previously defined. The term in \eqref{eq:short-ind-1}, $\Delta_{\rm ind}^{\rm s}$, admits a clean interpretation for each value $a\in \mathcal A$ under additional assumptions we introduce below. Without additional assumptions, this term is a type of ``indirect'' effect that contains the product of the difference in conditional probabilities, $\pi_{1}(a)-\pi_{0}(a)$, and a term that confounds the average partial causal effect of $A$ moving from $0$ to $a$ on $Y$, with selection that arises from the distinct conditioning sets $\{D=0,A=a\}$ and $\{D=0,A=0\}$. Finally, the term in \eqref{eq:short-sel-1}, $ \Delta_{\rm sel}^{\rm s}$, is a selection term that captures the fact that $Y(0,a)$ may not be independent of $(D,A)$. 

Three points are worth highlighting. First, the decomposition above does not rely on either Assumption \ref{ass:basic-SOO} or Assumption \ref{ass:joint-SOO}. While Assumption \ref{ass:basic-SOO} (ignoring the $X$) ensures that
\begin{equation}\label{eq:short_under_A31}
  \Delta_{\rm short} = E[Y|D=1]- E[Y|D=0] = E[Y(1)-Y(0)]~,
\end{equation}
where \( Y(d) \) denotes the pooled potential outcomes defined in \eqref{eq:pooled}, it is not sufficient to interpret \( \Delta_{\rm short} \) as an average direct causal effect or as a parameter satisfying strong sign preservation. In particular, the two endogeneity channels—\( \Delta_{\rm ind}^{\rm s} \) and \( \Delta_{\rm sel}^{\rm s} \)—in the decomposition of \( \Delta_{\rm short} \) may be positive or negative and, crucially, may cause \( \Delta_{\rm short} \) to have the opposite sign of \( \Delta_{\rm dce}^{\rm s} \). Second, these two endogeneity channels are conceptually distinct. The components in \( \Delta_{\rm ind}^{\rm s} \) reflect indirect pathways and are generally difficult to eliminate, while the selection term \( \Delta_{\rm sel}^{\rm s} \) becomes zero under Assumption \ref{ass:joint-SOO}. Finally, under a given set of assumptions, the decomposition we derive for \( \Delta_{\rm short} \)—as well as for the other estimands analyzed in subsequent sections—is unique up to the choice of conditioning set (in our case, \( \Omega = \{D=d, A=a\} \)) and the normalization (here, the indirect effect is defined relative to \( \mu(0,0) \)). Once Assumption \ref{ass:joint-SOO} is imposed, the conditioning event becomes irrelevant, and the only remaining degree of freedom lies in the choice of normalization for the indirect effect.

Under Assumption \ref{ass:joint-SOO} the three terms entering the decomposition for $\Delta_{\rm short}$ simplify in the following way: $\Delta_{\rm sel}^{\rm s} = 0$ and
\begin{align}
  \Delta_{\rm dce}^{\rm s} &= \sum_{a\in \mathcal A} \pi_{1}(a)  (\mu(1,a)-\mu(0,a))\label{eq:short-dce-2}\\
  \Delta_{\rm ind}^{\rm s} &= \sum_{a\in \mathcal A}(\pi_{1}(a)-\pi_{0}(a))(\mu(0,a)-\mu(0,0))\label{eq:short-ind-2}.
\end{align}
That is, the selection term $\Delta_{\rm sel}^{\rm s}$ is no longer present, and the direct effect $\Delta_{\rm dce}^{\rm s}$ and indirect effect $\Delta_{\rm ind}^{\rm s}$ are now a function of the unconditional expectations $\mu(d,a)$. Importantly, the term $\Delta_{\rm ind}^{\rm s}$ is still part of the decomposition since Assumption \ref{ass:joint-SOO} does not restrict how $A$ may affect outcomes, so that $\mu(0,a)-\mu(0,0)\ne 0$, nor does it affect how the main action may affect the other ones, so that $\pi_{1}(a)-\pi_{0}(a)\ne 0$. Aside from removing the term capturing selection bias, Assumption \ref{ass:joint-SOO} also delivers a clean interpretation of the indirect effects captured by $ \Delta_{\rm ind}^{\rm s}$. Each summand in $ \Delta_{\rm ind}^{\rm s}$ contains the average partial causal effect of $A$ moving from $0$ to $a$ on $Y$, $\mu(0,a)-\mu(0,0)$, multiplied by the difference $\pi_{1}(a)-\pi_{0}(a)$, which admits a causal interpretation of an average direct causal effect of $D$ on $A$ under the additional assumption $A(d)\perp D$. 

We can interpret the terms entering the decomposition of $\Delta_{\rm short}$ in \eqref{eq:short-decomposition} in the context of the examples we introduced in Section \ref{sec:notation}. For example, in the farming example where $Y$ is crop yield, $D$ is an indicator of the use of fertilizer, and $A$ is, for simplicity, an indicator of high water usage. In this setting, $\Delta_{\rm dce}^{\rm s}$ captures the average direct causal effect of using fertilizer on the crop yield, where the effect weights units with high and low water usage according to the respective probabilities of these actions happening for the treated, $\pi_1(a)$. The term $\Delta_{\rm ind}^{\rm s}$, in turn, captures a piece of the causal effect of water usage on crop yield that depends on the magnitude of differential water usage between the treated and the untreated. If water usage causally improves crop yield in the absence of fertilizer, and getting an exogenous fertilizer incentivizes units to increase their water usage, this term would be positive.  

The following theorem summarizes our discussion above.

\begin{theorem}\label{thm:short}
Consider the short regression in \eqref{eq:short-regression} and assume $P\{D=d,A=a\}>0$ for all $(d,a)\in \mathcal D\times \mathcal A$. Then, $\Delta_{\rm short}$ can be decomposed as in \eqref{eq:short-decomposition}-\eqref{eq:short-sel-1}. If Assumption \ref{ass:joint-SOO} holds, then $\Delta_{\rm sel}^{\rm s}=0$ and $\Delta_{\rm dce}^{\rm s}$ and $\Delta_{\rm ind}^{\rm s}$ simplify to the expressions in \eqref{eq:short-dce-2} and \eqref{eq:short-ind-2}. 
\end{theorem}

\begin{remark}\label{rem:short-nsp}
It is important to note that, even under the exogeneity condition in Assumption \ref{ass:joint-SOO}, the parameter $\Delta_{\rm short}$ does not satisfy strong sign preservation as defined in Definition \ref{def:strong-SP}. Indeed, it is certainly possible that $\mu(1,a)-\mu(0,a)>0$ for all $a \in \mathcal A$ and yet $\Delta_{\rm short}<0$ due to $\Delta_{\rm ind}^{\rm s}<-\Delta_{\rm dce}^{\rm s}<0$. This phenomenon, which is reminiscent of the Simpson's paradox, is present in similar ways in other causal inference settings; including two-way fixed effects as in \cite{DeChaisemartin/xavier:20} or the surrogate paradox as in \cite{vanderweele:13}.\footnote{As pointed out by a referee, strong sign preservation would hold under additional assumptions that ensure \( {\rm sign}( \Delta_{\rm ind}^{\rm s})= {\rm sign}(\Delta_{\rm dce}^{\rm s})\). Examples include monotonicity in \( \pi_{d}(a) \) and \( \mu(0,a) \), or linearity assumptions on \( \mu(d,a) \), as we discuss in detail in the next section.} \qed
\end{remark}

\begin{remark}\label{rem:total-vs-partial}
Under Assumption \ref{ass:joint-SOO} and $A(d)\perp D$, $\Delta_{\rm short}$ is a linear combination of average partial causal effects, and it captures a ``total'' effect rather than a ``partial'' effect, as discussed in Remark \ref{rem:ssp-comment}. To understand this, notice that $\Delta_{\rm ind}$ in \eqref{eq:short-ind-2} is the product of $\pi_{1}(a)-\pi_{0}(a)$ and $\mu(0,a)-\mu(0,0)$ for each $a\in\mathcal A$. Both of these terms are partial effects, where $\pi_{1}(a)-\pi_{0}(a)$ is the average partial effect of $D$ on $A$ and $\mu(0,a)-\mu(0,0)$ is the average partial effect of moving $A$ from $0$ to $a$ on the outcome $Y$ for units with $D=0$. With this interpretation, $\Delta_{\rm short}$ captures a total effect of $D$ on $Y$ that adds up the direct effect of $D$ on $Y$, captured by $\Delta_{\rm dce}^{\rm s}$, and the indirect effect that $D$ has on $Y$ via its effect on $A$ and how $A$ affects $Y$. This distinction between partial and total effects mimics the usual one associated with total and partial derivatives in mathematical analysis. Whether total or partial effects are most relevant will depend on the application at hand; see, for example, \citet{Manski:97,heckman:00,imai2010,glynn2012product}. That said, the premise of this paper is that regression-based estimands like $\Delta_{\rm short}$ are often used with the goal of recovering an average direct causal effect, as defined in Definition \ref{def:DCE}. Our objective is therefore to clarify the conditions under which such interpretations are valid. \qed
\end{remark}

Assumption \ref{ass:basic-SOO} suffices to interpret a difference-in-means as the total (average) treatment effect $E[Y(1) - Y(0)]$, but it is generally insufficient to remove selection effects–see Example \ref{ex:basic_with_selection}. Moreover, as emphasized in the empirical mediation literature, direct and indirect effects cannot be separately identified in randomized experiments without additional assumptions; see \citet{robins1992identifiability}. Similar recognition has emerged in applied work, particularly in development economics---for example, \citet{Mel/etal:2009}, \citet{Duflo/etal:2011}, and \citet{Beaman/etal:2013}. For these reasons, we prioritize results derived under Assumption \ref{ass:joint-SOO}, which enables clean interpretations of regression-based estimands in terms of direct and indirect causal effects.

\subsection{Long regression}\label{sec:long}
A seemingly natural, and certainly popular, way to mitigate the presence of indirect effects and obtain an estimand that satisfies strong sign preservation is to control for the other actions linearly as in \eqref{eq:long-regression}; an approach we called the long regression. Our main result below shows that the slope coefficient $\Delta_{\rm long}$ in \eqref{eq:long-regression} admits a decomposition similar to that derived by $\Delta_{\rm short}$, and thus includes a combination of direct effects and indirect effects. However, except in some special cases, the coefficients multiplying each average partial causal effect, as in Definition \ref{def:PCE}, could be negative and so $\Delta_{\rm long}$ may fail to satisfy strong sign preservation  even in the absence of indirect effects. We formalize this below and provide a proof in Appendix \ref{app:proofs}. 

\begin{theorem}\label{thm:long}
Let Assumption \ref{ass:joint-SOO} hold and assume that $P\{D=d,A=a\}>0$ for all $(d,a)\in \mathcal D\times \mathcal A$ and that the covariance matrix of $(D,A)$ is positive definite. Then, the coefficient $\Delta_{\rm long}$ in \eqref{eq:long-regression} admits the decomposition
\begin{equation}\label{eq:long-decomposition}
\Delta_{\rm long} = \Delta_{\rm dce}^{\rm l}+\Delta_{\rm ind}^{\rm l}~,
\end{equation}
where
\begin{align}
\Delta_{\rm dce}^{\rm l}& \equiv  \sum_{a\in \mathcal{A}}\omega^{\rm l}_{\rm dce}(a)(\mu (1,a)-\mu (0,a)) \label{eq:long-dce} \\
\Delta_{\rm ind}^{\rm l}& \equiv  \sum_{a\in \mathcal{A}}\omega^{\rm l}_{\rm ind}(a)(\mu (0,a)-\mu (0,0))~,\label{eq:long-ind}
\end{align}
and $\{\omega^{\rm l}_{\rm dce}(a):a\in \mathcal{A}\}$ and $\{\omega^{\rm l}_{\rm ind}(a):a\in \mathcal{A}\}$ are as defined in Theorem \ref{thm:long_pre} and satisfy $\sum_{a\in \mathcal{A}}\omega^{\rm l}_{\rm dce}(a)=1$ and $\sum_{a\in \mathcal{A}}\omega^{\rm l}_{\rm ind}(a)=0$. Furthermore, the following statements are equivalent:
\begin{enumerate}[(a)]
\item $A$ are mutually exclusive binary variables, i.e., $\mathcal{A}_{j}=\{ 0,1\} $ for $j=1,\ldots ,K$ and $A_{j}A_{l}=0$ for all $j,l=1,\ldots ,K$ with $j\not=l$.
\item For any distribution of $( A,D) $, $\omega^{\rm l}_{\rm dce}(a)\geq 0$ for all $a\in \mathcal{A}$.
\item For any distribution of $( A,D) $, $\omega^{\rm l}_{\rm ind}(a)=0$ for all $a\in\mathcal A$.
\item For any distribution of $( A,D) $, $\Delta_{\rm long}$ satisfies strong sign preservation.
\end{enumerate}
\end{theorem}

Theorem \ref{thm:long} shows that $\Delta_{\rm long}$ can be decomposed into direct and indirect effects, but it leaves open the possibility that the coefficients entering each of these terms could, in general, be negative. An immediate implication is that, except in the special case where the actions in $A$ are all mutually exclusive binary variables, which includes the case where $A$ is a scalar binary variable as a special case, the term $\Delta_{\rm dce}^{\rm l} $ could be negative even if $\mu(1,a)-\mu(0,a)>0$ for all $a\in\mathcal A$. This is because $\omega _{\rm dce}(a)$ may be negative for some $a\in \mathcal{A}$. 
As a result, $\Delta_{\rm long}$ generally does not satisfy strong sign preservation for the following two reasons. First, it may be possible that $\Delta_{\rm ind}^{\rm l}<-\Delta_{\rm dce}^{\rm l}$, so that the indirect effect dominates the direct effect. This phenomenon is the same as the one we discussed for the short regression. Second, even in the absence of indirect effects, where $\Delta_{\rm ind}^{\rm l}=0$, the term $\Delta_{\rm dce}^{\rm l}$ could be negative by itself even if $\mu(1,a)-\mu(0,a)>0$ for all $a\in\mathcal A$, due to $\omega _{\rm dce}(a)<0$ for some $a\in\mathcal A$. This second possibility represents a stark distinction between the long regression estimand, $\Delta_{\rm long}$, and the short regression estimand, $\Delta_{\rm short}$, in the sense that $\Delta_{\rm long}$ does not even measure a total causal effect of $D$ on $Y$ without additional assumptions, cf.\ Remark \ref{rem:total-vs-partial}. It is also important to emphasize that this result does not depend on the distribution of $Y$ given $(A, D)$. As a consequence, both $\Delta_{\rm long}$ and $\Delta_{\rm dce}^{\rm l}$ can be arbitrarily negative—even if $\mu(1,a) - \mu(0,a) > 0$ for all $a \in \mathcal{A}$—when the actions are not mutually exclusive and binary. 

We emphasize that the results in Theorem \ref{thm:long} are all \emph{equivalent}. Obtaining an easy-to-interpret characterization of $\Delta_{\rm long}$ in terms of partial effects when the actions are all binary and mutually exclusive should not come as particularly surprising given related results on the properties of this estimand in different but related settings \citep[see, for example,][]{zhao/ding:22,goldsmithpinkham/etal:22}. To the best of our knowledge, however, a novel lesson of Theorem \ref{thm:long} is that this is precisely \emph{the only way} to guarantee a proper weighted average representation for $\Delta_{\rm long}$ for any distribution of $( A,D)$. This result could also be interpreted in the context of ``bad controls'' discussed in \cite{angrist2008mostly}, as it clarifies the unique way in which seemingly bad controls could actually become ``good controls''.

\begin{remark}\label{rem:strong-to-basic-1}
  The role of Assumption \ref{ass:joint-SOO} in our decompositions is purely to simplify the expressions and facilitate their interpretation. Removing this assumption leads to a decomposition of $\Delta_{\rm long}$ that differs from Theorem \ref{thm:long} in three key ways. First, the term $\Delta_{\rm dce}^{\rm l}$ becomes a linear combination of expectations that condition on $\Omega = \{D=1, A=a\}$. Second, the interpretation of $\Delta_{\rm ind}^{\rm l}$ becomes more convoluted, for the same reasons discussed earlier in the context of $\Delta_{\rm ind}^{\rm s}$. Third, the decomposition includes an additional selection term that is conceptually identical to $\Delta_{\rm sel}^{\rm s}$ in the short regression case. Moreover, without Assumption \ref{ass:joint-SOO}, the equivalence between (a), (b), and (c) in Theorem \ref{thm:long} remains valid. However, in this case, the selection effect that appears in the decomposition breaks down the equivalence between (d) and the other items. For details on the decomposition of $\Delta_{\rm long}$ without imposing Assumption \ref{ass:joint-SOO}, see Theorem \ref{thm:long_pre} in Appendix \ref{app:proofs}. \qed
\end{remark}

The possibility of \( \omega^{\rm l}_{\rm dce}(a) \) being negative for some \( a \in \mathcal A \) does not depend on pathological data-generating processes. It can occur in relatively simple settings with reasonable distributions for \( (A,D) \). This raises a significant concern for the use of linear-in-\( A \) regressions, as they can lead to results that are difficult to interpret and, in many cases, offer no improvement over the short regression in \eqref{eq:short-regression}. Below we illustrate this situation with two canonical simple cases: one where $A$ is a scalar random variable taking multiple values, and another one where $A = (A_1,A_2)$ with $A_1$ and $A_2$ being binary. 

Consider first the case where $A=A_1$ is a scalar random variable taking values in $\mathcal A_1 = \{0,1,2,\dots,\bar{\rm a}_1\}$. The regression in \eqref{eq:long-regression} simplifies to 
\begin{equation*}\label{eq:long-regression-scalar}
  Y = \Delta_{\rm long} D + \theta_0 + \theta_1 A_1 + V~.
\end{equation*}
Theorem \ref{thm:long_pre} provides general closed-form expressions for $\{\omega^{\rm l}_{\rm dce}(a):a\in\mathcal A\}$ and $\{\omega^{\rm l}_{\rm ind}(a):a\in\mathcal A\}$ that, when applied to this specific example, lead to   
\begin{equation}\label{eq:omega-d-a-1}
    \omega^{\rm l}_{\rm dce}(a) \propto \left(\pi_1(a)- \frac{\cov(D,A_1)(a-E[A_1])}{\var(A_1)(1-p)} \right) ~,
\end{equation}
where $p=P\{D=1\}$. From this expression, it follows that any distribution of $(A,D)$ for which 
$\cov(D,A_1)(a-E(A_1))>\var(A_1)(1-p)\pi_1(a)$, would lead to negative weights. For example, consider the case where $p=0.8$, $\bar{\rm a}_1=3$, $\{ A|D=1\} \sim {\rm Bi}(3,0.8) $, and $\{ A|D=0\} \sim {\rm Bi}(3,0.2) $, where ${\rm Bi}(n,\pi)$ denotes a Binomial distribution with $n$ trials and success probability $\pi$. In this case, $\omega^{\rm l}_{\rm dce}(3)=-0.41<0$. Figure \ref{fig:weights} plots the weights $\omega_{\rm dce}^{\rm l}(a)$ as a function of $p$ and shows that $\omega^{\rm l}_{\rm dce}(3)$ is negative for any $p>0.4$ in this example. 

\begin{figure}[b!]
\centering
 
\pgfplotsset{ 
cycle list={
{draw=black, solid,very thick,orange!50!red},
{draw=black,dotted,very thick,orange!70!red}, 
{draw=black,densely dashed,very thick,orange!90!red},
{draw=black, dashdotted, very thick,black},
{only marks, mark=asterisk}}}

    \begin{tikzpicture}[baseline]
         \begin{axis}[width=4.2in,height = 2.5in,xmin=0.1, xmax=0.9,ymin=-0.6,ymax=1,ytick={-0.5,-0.25,0,0.25,0.5,0.75,1},xlabel={$P\{D=1\}$},ylabel={$\omega^{\rm l}_{\rm dce}(a)$}]
              \addplot  table[y = {w1}]{plotdata.dat};
              \addplot  table[y = {w2}]{plotdata.dat};
              \addplot  table[y = {w3}]{plotdata.dat};
              \addplot  table[y = {w4}]{plotdata.dat};          
         \end{axis}
    \end{tikzpicture}

    \begin{tikzpicture}
        \begin{customlegend}[legend columns=4,legend style={align=center,column sep=3ex},legend entries={{$\omega^{\rm l}_{\rm dce}(0)$},{$\omega^{\rm l}_{\rm dce}(1)$}, {$\omega^{\rm l}_{\rm dce}(2)$},{$\omega^{\rm l}_{\rm dce}(3)$}}]
         \addlegendimage{solid,very thick,orange!50!red}
        \addlegendimage{dotted,very thick,orange!70!red}
        \addlegendimage{densely dashed,very thick,orange!90!red}   
        \addlegendimage{dashdotted, very thick,black}
        \end{customlegend}
    \end{tikzpicture}

    \caption{\footnotesize{Weights $\omega_{\rm dce}^{\rm l}(a)$ as a function of $p$ when $\{ A|D=1\} \sim {\rm Bi}(3,0.8) $, and $\{ A|D=0\} \sim {\rm Bi}(3,0.2) $}}
\label{fig:weights}
\end{figure}
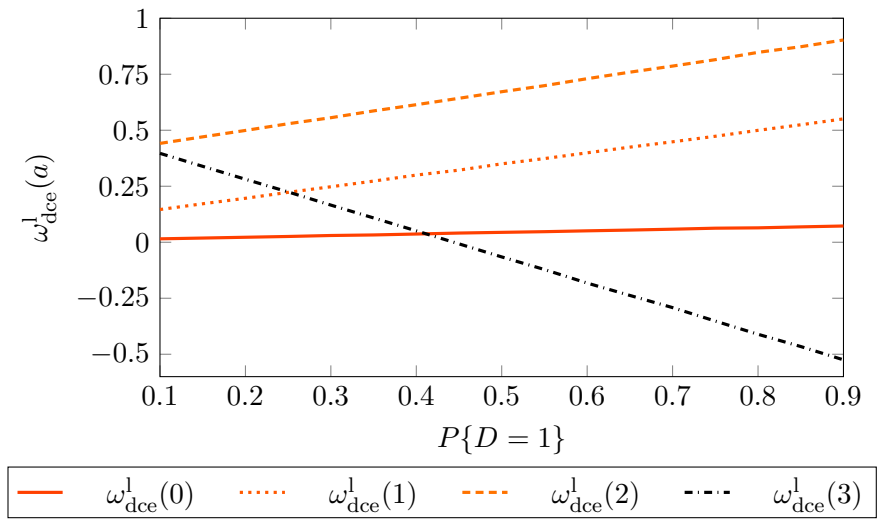

Next, consider the case where $A=(A_1,A_2)$ with $A_1$ and $A_2$ both being binary variables, so that $\mathcal A = \{(0,0),(1,0),(0,1),(1,1)\}$. The regression in \eqref{eq:long-regression} simplifies to 
\begin{equation*}
  Y = \Delta_{\rm long} D + \theta_0 + \theta_1 A_1 + \theta_2 A_2 + V~.
\end{equation*}
The closed-form expressions for $\{\omega^{\rm l}_{\rm dce}(a):a\in\mathcal A\}$ and $\{\omega^{\rm l}_{\rm ind}(a):a\in\mathcal A\}$ derived in Theorem \ref{thm:long_pre} also simplify to this case and lead to simple conditions for which   
\begin{equation}\label{eq:opposite-weights}
  \omega^{\rm l}_{\rm dce}(1,1) = - \omega^{\rm l}_{\rm dce}(1,0)~.
\end{equation}
That is, whenever one of the average partial causal effects gets a positive weight, the other necessarily gets a negative one. As an illustrative example, when $\cov(A_1,A_2)=0$,
\begin{equation}
  P\{D=1\}=P\{A_2=1\} = \frac{1}{2}, \quad P\{A_1=1\mid D=1\} = 2P\{A_1=1\}~,
\end{equation}
and 
\begin{equation}
  P\{A_1=A_2=1 \mid D=1\} = P\{A_1=1,A_2=0 \mid D=1\}=\frac{1}{4}~. 
\end{equation}
Using the expressions in Theorem \ref{thm:long_pre}, we obtain $\omega^{\rm l}_{\rm dce}(1,0)=-\omega^{\rm l}_{\rm dce}(1,1)=-0.30$, which one more time illustrates that negative weights arise naturally in settings with non-pathological DGPs. In the proof of Theorem \ref{thm:long}, we present even simpler counter-examples that also illustrate how the weights $\{\omega^{\rm l}_{\rm ind}(a):a\in\mathcal A\}$ are generally non-zero and potentially negative, as well as how $\omega^{\rm l}_{\rm dce}(a)$ may be negative without necessarily satisfying \eqref{eq:opposite-weights}.

\begin{remark}\label{rem:long-nsp}
It is important to note that, even under the stronger exogeneity condition in Assumption \ref{ass:joint-SOO}, the parameter $\Delta_{\rm long}$ does not generally satisfy strong sign preservation as defined in Definition \ref{def:strong-SP}. Indeed, it is certainly possible that $\mu(1,a)-\mu(0,a)>0$ for all $a \in \mathcal A$ and yet $\Delta_{\rm long}<0$ due to either $\Delta_{\rm ind}^{\rm l}<-\Delta_{\rm dce}^{\rm l}<0$ or simply $\Delta_{\rm dce}^{\rm l}<0$ because of negative weights $\{\omega^{\rm l}_{\rm dce}(a):a\in\mathcal A\}$. This second condition implies that \emph{not even} $\Delta_{\rm dce}^{\rm l}$ satisfies strong sign preservation and thus, in general, $\Delta_{\rm long}$ does not offer much of a benefit relative to $\Delta_{\rm short}$, as $\Delta_{\rm short}$ can at least be interpreted as a total effect, as discussed in Remark \ref{rem:total-vs-partial}.  \qed
\end{remark}

\begin{remark}\label{rem:applications-long}
As discussed in Remark \ref{rem:glynn}, the long regression in \eqref{eq:long-regression} is used extensively in the social sciences and the mediation literature. In economics, \citet[][Eq.\ (6)]{heckman2013understanding} consider a long regression in the context of a more restrictive model for potential outcomes that is linear and separable in $(d,a)$. More recently, \citet[][Eq.\ (7)]{fagereng2021wealthy} use the same mediation model from \cite{heckman2013understanding}, in combination with the long regression in \eqref{eq:long-regression}, to disentangle the average causal effect on outcomes into direct and indirect effects. The ``direct'' causal interpretation assigned to the estimands in this last set of papers is correct under the modeling assumptions for potential outcomes, despite both applications involving actions $A$ that are multidimensional and non-mutually exclusive. Our results, however, imply that the main conclusions from such an analysis delicately rely on a linear model for $\mu(d,a)$ and do not generally extend to nonlinear ones.\qed
\end{remark}

The results in Theorem \ref{thm:long} are, to the best of our knowledge, novel—though they are certainly related in spirit to a large literature in biostatistics and epidemiology concerned with mediators and indirect effects. Given the breadth of that literature, we highlight a few representative contributions rather than attempting an exhaustive review.\footnote{Two illustrative examples that were suggested by a referee are \citet{groenwold/etal:2021} and \citet{wilcox/weinberg/basso:2011}. The first paper presents a decompositions of certain regression estimands similar in form to ours, but do so under the assumption that all variables are linearly related to one another—an assumption that, as we show below, has strong implications for interpretation. The second paper warns against including certain variables (such as gestational age) as controls in regression analyses, noting the potential for collider bias. Their critique is motivated by a specific empirical context and supported by simulation, but they do not offer a formal decomposition nor do they analyze the unadjusted case (analogous to our short regression).} While our results are thematically aligned with several ideas in this literature, we are not aware of prior work that establishes the same characterizations with the level of formality and generality presented here. A key contribution of our paper is to show that the case in which the components of $A$ are binary and mutually exclusive is essentially the \emph{only} one in which the long regression can be interpreted as capturing direct causal effects—without requiring additional assumptions such as linearity. This insight, which is formalized in Theorem \ref{thm:long}, sharpens and unifies a number of informal arguments that appear across distinct literatures.

Perhaps the most relevant work for our purposes is \citet{imai2010}, who study the interpretation of the long regression—originally popularized by \citet{baron1986moderator}—under a linear potential outcome model and an assumption they call sequential ignorability. We restate that assumption below to facilitate comparison.

\begin{assumption}[Sequential Ignorability]\label{ass:seq-igno}
Let $A(d)$ denote the potential outcome for $A$ and assume that: (i) $(Y(d',a),A(d)) \perp D | X=x$, and (ii) $Y(d',a)\perp A(d) | D=d, X=x$, both for $d,d'=0,1$ and all $x$. 
\end{assumption}
The results in \cite{imai2010} about the long regression in \eqref{eq:long-regression} invoke (a) sequential ignorability, (b) a scalar random variable $A$ (though not necessarily binary), and (c) a linear model for $\mu(d,a)$ in $(d,a)$. Under these conditions (a)-(c), and ignoring the $X$ for simplicity, \citet[][Theorem 2]{imai2010} shows that $\Delta_{\rm long}$ identifies $\bar\zeta=\bar\zeta(1)=\bar\zeta(0)$ where $\bar\zeta(d) \equiv E[Y(1,A(d))] - E[Y(0,A(d))] =  \sum_{a\in \mathcal A}(\mu(1,a)-\mu(0,a))\pi_{d}(a)$, and the equality follows from Assumption \ref{ass:seq-igno} implying $ Y(d',a)\perp A(d)$; see Lemma \ref{lem:linear-mu} in Appendix \ref{app:aux-lemmas}. The linear model for $\mu(d,a)$ implies that the difference $\mu(1,a)-\mu(0,a)$ does not depend on the value of $a$, and so it is just a constant that we can denote by $\bar \zeta$ without loss of generality. The fact that  $\bar\zeta=\bar\zeta(1)=\bar\zeta(0)$ then follows from $\sum_{a\in \mathcal A}\pi_{d}(a)=1$ for $d\in \{0,1\}$. 

The additional assumptions (a)-(c) mentioned above have implications on the conclusions of Theorem \ref{thm:long}, which does not invoke any of these assumptions. In particular, the linearity of $\mu(d,a)$ implies that $\Delta^{\rm l}_{\rm dce}$ in \eqref{eq:long-dce} equals $\bar \zeta \sum_{a\in\mathcal A}\omega^{\rm l}_{\rm dce}(a)=\bar \zeta$ by the weights adding up to one according to Theorem \ref{thm:long}. The same linearity assumption also implies that 
$$ \Delta_{\rm ind}^{\rm l} \equiv  \sum_{a\in \mathcal{A}}\omega^{\rm l}_{\rm ind}(a)(\mu (0,a)-\mu (0,0)) \propto  \sum_{a\in \mathcal{A}}\omega^{\rm l}_{\rm ind}(a)a = 0~,$$
where the last equality follows from Theorem \ref{thm:long_pre} in the appendix. We conclude that Theorem \ref{thm:long} coincides with the results in \cite{imai2010} in delivering $\Delta_{\rm long}$ being equal to $\bar\zeta$ under the additional assumption that $\mu(d,a)$ is linear in $(d,a)$. This means that, while sequential ignorability is a stronger assumption than Assumption \ref{ass:joint-SOO} (we prove this claim in Lemma \ref{lem:seq-igno} in Appendix \ref{app:aux-lemmas}), the main driving force of this result is the linear model for the potential outcomes or, equivalently, the linear model for $\mu(d,a)$.\footnote{Sequential ignorability without a linear model has also been used in the literature on causal mediation analysis to construct semi-parametrically efficient estimators; see \cite{tchetgen/etal:12}.} As we discussed in Remark \ref{rem:applications-long}, this linearity assumption has been used in economic applications, e.g., \cite{heckman2013understanding,fagereng2021wealthy}, and while it imposes enough restrictions to provide a clean interpretation to the coefficient $\Delta_{\rm long}$, our results show that such a clean interpretation generally breaks down when $\mu(d,a)$ is not linear in $(d,a)$.

\begin{remark}\label{rem:Imai-indirect-effects}
We note that while Theorem \ref{thm:long} is a result on how to properly interpret $\Delta_{\rm long}$ in the context of a long regression, \citet[][Theorem 2]{imai2010} is a result on the identification of \emph{natural} indirect effects via the same type of regression and the additional conditions (a)-(c) above. Related to our discussion in Remark \ref{rem:controlled-effects}, the \emph{natural} indirect effect does not coincide with our notion of indirect effect in Theorem \ref{thm:long}. To see the difference, note that the \emph{natural} indirect effect, defined as $\bar\delta(d) = E[Y(d,A(1)) - Y(d,A(0))]$, can be written as  
\begin{equation}\label{eq:bar-delta}
   \bar \delta(d) = \sum_{a\in \mathcal A}(\mu(d,a)-\mu(d,0))(\pi_{1}(a) - \pi_{0}(a))~,  
\end{equation}
under sequential ignorability, and is distinct from $\Delta^{\rm l}_{\rm ind}$ in \eqref{eq:long-ind} because $\omega^{\rm l}_{\rm ind}(a) \ne \pi_{1}(a) - \pi_{0}(a)$. Conceptually, the literature on mediation analysis defines an indirect effect as a target parameter and then determines conditions under which such indirect effects could be identified from the data. In contrast, we characterize the decomposition of estimands in terms of average direct causal effects, as defined in Definition \ref{def:DCE}, and then group the remaining terms as indirect or selection terms, depending on the case. We note, however, that our indirect effects coincide with those characterized by $\bar \delta(d)$ in the case of the short regression, i.e., $\Delta^{\rm s}_{\rm ind}$ in \eqref{eq:short-ind-2} equals $\bar \delta(0)$. \qed
\end{remark}

\begin{remark}
Under the conditions of Theorem \ref{thm:long} and assuming that $\mu(d,a)$ is linear, it follows that $\Delta_{\rm ind}^{\rm l} = 0$, and so $\Delta_{\rm long}$ captures the direct causal effect. In this case, $\Delta_{\rm short}$ equals the direct effect plus the indirect effect, justifying the common practice of using $\Delta_{\rm short} - \Delta_{\rm long}$ to estimate the indirect effect. However, Theorems \ref{thm:short} and \ref{thm:long} show that this interpretation fails when $\mu(d,a)$ is nonlinear. Then, (i) the direct effects identified by the two regressions may differ, and (ii) $\Delta_{\rm long}$ may itself include an indirect component if actions are not mutually exclusive and binary. As a result, $\Delta_{\rm short} - \Delta_{\rm long}$ no longer isolates the indirect effect. \qed
\end{remark}

\begin{remark}\label{rem:kaufman04}
Our findings conceptually align with prior critiques in the literature on mediation analysis regarding the limitations of interpreting regression coefficients as causal direct effects. In particular, \citet{kaufman2004critique} provide an early warning against using covariate adjustment to infer biological mediation. Using a potential outcomes framework, they show that even in randomized experiments, adjusting for post-treatment variables can yield misleading inferences unless strong and often unrealistic assumptions hold—such as no unmeasured mediator-outcome confounding or no interaction at the unit level. \citet{kaufman2009gilding} further advocates for grounding mediation analysis in formal causal estimands, such as controlled direct effects, rather than relying on statistical associations. These contributions reinforce our argument that regression-based estimands like $\Delta_{\rm short}$ and $\Delta_{\rm long}$ may not recover interpretable causal effects in the absence of strong identifying assumptions. \qed
\end{remark}


\subsection{Long regression with interactions}\label{sec:longInteractions}

A common variant of the long regression additionally includes interactions between the \(K\) actions, \(A_{1}, \dots, A_{K}\), and the treatment \(D\); that is, the slope coefficient \(\Delta_{\rm inter}\) in \eqref{eq:inter-regression}. We refer to this as the long regression with interactions. Our result below shows that the slope coefficient \(\Delta_{\rm inter}\) in \eqref{eq:inter-regression} admits a decomposition with the same shortcomings as the one we derived for \(\Delta_{\rm long}\), including the possibility of \(\Delta_{\rm inter}\) not satisfying strong sign preservation, even in the absence of indirect effects. We formalize this below and provide a proof in Appendix \ref{app:proofs}.

\begin{theorem}\label{thm:inter} 
Let Assumption \ref{ass:joint-SOO} hold and assume that $P\{D=d,A=a\}>0$ for all $(d,a)\in \mathcal D\times \mathcal A$ and that the covariance matrix of $(D,A,AD)$ is positive definite. Then, the coefficient $\Delta_{\rm inter}$ in \eqref{eq:inter-regression} admits the decomposition
\begin{equation}\label{eq:inter-decomposition}
\Delta_{\mathrm{inter}} = \Delta_{\rm dce}^{\rm i}+\Delta_{\rm ind}^{\rm i}~,
\end{equation}
where 
\begin{align*}
\Delta_{\rm dce}^{\rm i}& \equiv \sum_{a\in \mathcal{A}}\omega^{\rm i} _{\rm dce}(a)(\mu (1,a)-\mu (0,a)) \\
\Delta_{\rm ind}^{\mathrm{i}}& \equiv \sum_{a\in \mathcal{A}}\omega^{\rm i} _{\rm ind}(a)(\mu (0,a)-\mu (0,0))~,
\end{align*}
and $\{\omega^{\rm i}_{\rm dce}(a):a\in \mathcal{A}\}$ and $\{\omega^{\rm i}_{\rm ind}(a):a\in \mathcal{A}\}$ are as defined in Theorem \ref{thm:inter_pre} and satisfy $\sum_{a\in \mathcal{A}}\omega^{\rm i}_{\rm dce}(a)=1$ and $\sum_{a\in \mathcal{A}}\omega^{\rm i}_{\rm ind}(a)=0$. Furthermore, the following statements are equivalent:
\begin{enumerate}[(a)]
\item $A$ are mutually exclusive binary variables, i.e., $\mathcal{A}_{j}=\{0,1\} $ for $j=1,\ldots ,K$ and $A_{j}A_{l}=0$ for all $j,l=1,\ldots ,K$ with $j\not=l$.
\item For any distribution of $( A,D) $, $\omega^{\rm i}_{\rm dce}(a)\geq 0$ for all $a\in \mathcal{A}$.
\item For any distribution of $(A,D)$, $\omega^{\rm i}_{\rm ind}(a)=0 $ for all $a\in \mathcal{A}$.
\item For any distribution of $(A,D)$, $\Delta_{\mathrm{inter}}$ satisfies strong sign preservation.
\end{enumerate}
\end{theorem}

Theorem \ref{thm:inter} is analogous to Theorem \ref{thm:long} and has very similar implications. Except in the special case where the actions in $A$ are all mutually exclusive binary variables, which includes the case where $A$ is a scalar binary variable as a special case, the term $\Delta_{\rm dce}^{\rm i} $ could be negative even if $\mu(1,a)-\mu(0,a)>0$ for all $a\in\mathcal A$. This is because $\omega^{\rm i}_{\rm dce}(a)$ may be negative for some $a\in \mathcal{A}$.\footnote{The possibility of \( \omega^{\rm i}_{\rm dce}(a) \) being negative for some \( a \in \mathcal A \) can occur in simple settings with reasonable distributions for \( (A,D) \); see the proof of Theorem \ref{thm:inter}.} As a result, $\Delta_{\rm inter}$ generally does not satisfy strong sign preservation for the same two reasons that $\Delta_{\rm long}$ fails to do so. First, it is possible that $\Delta_{\rm ind}^{\rm i} < -\Delta_{\rm dce}^{\rm i}$, and second, even if $\Delta_{\rm ind}^{\rm i} = 0$, the term $\Delta_{\rm dce}^{\rm i}$ itself can be negative—despite $\mu(1,a) - \mu(0,a) > 0$ for all $a \in \mathcal{A}$—due to negative weights. Again, this second failure mode distinguishes $\Delta_{\rm long}$ and $\Delta_{\rm inter}$ from $\Delta_{\rm short}$. 

While Theorem \ref{thm:inter} focuses on the properties of the estimand $\Delta_{\rm inter}$ in settings with interactions terms, it is most often the case that the analyst would rather focus on the estimand $\Delta_{\rm inter} + \sum_{j=1}^K \lambda_j E[A_j] $ (or, simply, $\Delta_{\rm inter} + \sum_{j=1}^K \lambda_j a_j $ for given values $a_j$, $j=1,\dots,K$). In Lemma \ref{lem:inter-w-lambda} in Appendix \ref{app:aux-lemmas}, we show that 
\begin{equation*}
\Delta_{\rm inter}+\sum_{j=1}^K \lambda_j E[A_j] = \sum_{a\in \mathcal{A}}\omega_{\mathrm{dce}}^{\mathrm{i}\star}(a)(E[Y(1,a)-Y(0,a)])+ \omega _{\mathrm{ind}}^{\mathrm{i\star}}(a) (E[ Y( 0,a) -Y( 0,0) ])~,
\end{equation*}
where the ``weights'' $\{(\omega_{\mathrm{dce}}^{\mathrm{i}\star}(a),\omega_{\mathrm{ind}}^{\mathrm{i}\star}(a)):a\in\mathcal A\}$ may be negative in general, and thus leading to an estimand with similar properties to those of $\Delta_{\rm inter}$. The one special case where the estimand $\Delta_{\rm inter} + \sum_{j=1}^K \lambda_j E[A_j] $ identifies the intended causal effects is when $\mu(d,a)$ is assumed to take the functional form $\mu(d,a) = \kappa_0 + \kappa_1 d +  \sum_{j=1}^K \kappa_{2,a}a_j +  d\sum_{j=1}^K \kappa_{3,j} a_j$, which is equivalent to assuming that the conditional mean of the observed outcome, $Y$, is correctly specified in the interaction regression in \eqref{eq:inter-regression}. Lemma \ref{lem:linear-mu-inter} in Appendix \ref{app:aux-lemmas} shows that $\Delta_{\rm inter} + \sum_{j=1}^K \lambda_j a_j = \mu(1,a) - \mu(0,a)$ in this case, delivering an average partial causal effect given $a=(a_1,\dots,a_K)$, as defined in Definition \ref{def:PCE}. It follows from these results that a clean interpretation of $\Delta_{\rm inter}$ or $\Delta_{\rm inter} + \sum_{j=1}^K \lambda_j a_j$ in terms of the definitions introduced in Section \ref{sec:causal-effects} essentially depends on a correctly specified parametric model for potential outcomes and does not generally apply to nonlinear models, similarly to our results about the long regression in Section \ref{sec:long}.

\begin{remark}\label{rem:strong-to-basic-2}
The observations made in Remark \ref{rem:strong-to-basic-1} for $\Delta_{\rm long}$ also extend to $\Delta_{\rm inter}$. In particular, one can derive a decomposition of $\Delta_{\rm inter}$ without imposing Assumption \ref{ass:joint-SOO}. Dropping this assumption leads to a decomposition that differs from Theorem \ref{thm:inter} in three key ways. First, the term $\Delta_{\rm dce}^{\rm i}$ becomes a linear combination of expectations conditional on ${D=1, A=a}$. Second, the interpretation of $\Delta_{\rm ind}^{\rm i}$ becomes convoluted, for the same reasons discussed in the context of $\Delta_{\rm ind}^{\rm s}$ and $\Delta_{\rm ind}^{\rm l}$. Finally, the decomposition includes an additional selection term that is conceptually identical to those appearing in $\Delta_{\rm ind}^{\rm s}$ and $\Delta_{\rm ind}^{\rm l}$. As with the long regression case, the equivalence between (a), (b), and (c) in Theorem \ref{thm:inter} remains valid even without Assumption \ref{ass:joint-SOO}. However, in the absence of this assumption, $\Delta_{\rm inter}$ includes the selection term, which breaks the equivalence between (d) and the other items. For details on the decomposition of $\Delta_{\rm long}$ without Assumption \ref{ass:joint-SOO}, see Theorem \ref{thm:inter_pre} in Appendix \ref{app:proofs}. \qed
\end{remark}

\begin{remark}\label{rem:applications-inter}
Similar to the long regression in \eqref{eq:long-regression} that we discussed in Remarks \ref{rem:glynn} and \ref{rem:applications-long}, the interaction regression is used extensively in the mediation literature. In the context of mediation analysis, this variant has been popularized and advocated by \cite{judd/kenny:81,kraemer/etal:2002,kraemer:etal:08}. However, the main goal in that particular setting has been to test for the existence of mediation effects using the estimated coefficients in \eqref{eq:inter-regression}, see \cite{kraemer:etal:08} for details on the proposed test and \cite{imai2010} for a result that shows that, under Assumption \ref{ass:seq-igno}, such a test does not provide evidence in favor or against the parameter $\bar \delta(d)$ in \eqref{eq:bar-delta} being zero. Our results, on the other hand, imply that even in settings where mediation effects are a nuisance and the main goal is to interpret the coefficients directly related to the treatment $D$, the main conclusions depend on the distribution of $(A,D)$. \qed
\end{remark}

\subsection{Strata fixed effects (SFE) regression}\label{sec:sfe}
Theorems \ref{thm:long} and \ref{thm:inter} show that adding other actions linearly in the regression is generally useful only when the actions are binary and mutually exclusive. If we could make the actions binary and mutually exclusive, we would obtain an estimand free from indirect effects and with strong sign preservation. This is possible by considering the slope coefficient \(\Delta_{\rm sfe}\) in \eqref{eq:sfe-regression}, where the regression controls for all possible values of \(A\), i.e., \(\{I\{A=a\}:a\in\mathcal A\}\). We call this a strata-fixed effects regression, due to its connection to the use of strata fixed effects in randomized controlled trials with covariate adaptive randomization; see \cite{bugni/canay/shaikh:2018,bugni/canay/shaikh:2019}. Our result below shows that \(\Delta_{\rm sfe}\) in \eqref{eq:sfe-regression} has a decomposition free from indirect effects (see Appendix \ref{app:proofs} for a proof).

\begin{theorem}\label{thm:sfe}
Let Assumption \ref{ass:joint-SOO} hold and assume that $P\{D=d,A=a\}>0$ for all $(d,a)\in \mathcal D\times \mathcal A$. Let $\pi_d(a)$ be defined as in \eqref{eq:pis}. Then $$\Delta_{\rm sfe}= \sum_{a\in\mathcal A}  \omega_{\rm sfe}(a)(\mu(1,a)-\mu(0,a))~,$$ where the weights $\{\omega_{\rm sfe}(a): a\in \mathcal A\}$ are given by  
\begin{equation}\label{eq:weights-sfe}
  \omega_{\rm sfe}(a) \equiv \frac{\pi_{1}(a)\pi_{0}(a)/p_a}{\sum_{a'\in \mathcal A} \pi_{1}(a')\pi_{0}(a')/p_{a'}}~
\end{equation}
for $p_a\equiv P\{A=a\}$, and satisfy $\sum_{a\in \mathcal{A}}\omega _{\mathrm{sfe}}(a)=1$ and $\omega _{\mathrm{sfe}}(a)\geq 0$. Furthermore, $\Delta_{\rm sfe}$ satisfies strong sign preservation.
\end{theorem}

Theorem \ref{thm:sfe} shows that $\Delta_{\rm sfe}$ identifies an average direct causal effect as in Definition \ref{def:DCE}. Importantly, it does not contain indirect effects and, as a result, $\Delta_{\rm sfe}$ satisfies strong sign preservation as in Definition \ref{def:strong-SP}. The weight $\omega_{\rm sfe}(a)$ admits a simple representation and depends only on the conditional probabilities that the action $a$ happens for the treated and control group, $\pi_{1}(a)$ and $\pi_{0}(a)$. These weights are generally different than the weights associated with the direct effect in the short regression, $\Delta_{\rm dce}^{\rm s}$, which are simply $\pi_{1}(a)$, unless $D$ and $A$ are independent. We emphasize that the result in Theorem \ref{thm:sfe} does not require the other actions, $A$, to be singled-valued or mutually exclusive. 

The estimand $\Delta_{\rm sfe}$ has been studied in other settings. For example, \cite{angrist:98} studies regressions of $Y$ on $D$ under Assumption \ref{ass:basic-SOO} and considers a linear regression that saturates on the covariates, $X$, to obtain an expression that parallels the one in Theorem \ref{thm:sfe}; see \citet[][footnote 11]{angrist:98}. Indeed, Assumption \ref{ass:joint-SOO} implies that $Y(d,a)\perp D|A$ which makes the connection to \cite{angrist:98} immediate. In this sense, our result illustrates another instance where stratifying confounding variables leads to easy-to-interpret results (though, it is important to notice that this result relies on the fact that the treatment variable $D$ in our case is binary, c.f. \cite{goldsmithpinkham/etal:22}). As another example, \cite{bugni/canay/shaikh:2018,bugni/canay/shaikh:2019} study the properties of this estimand and present results on how to properly compute standard errors in randomized controlled experiments with covariate adaptive randomization. These papers, however, do not represent $\Delta_{\rm sfe}$ as a weighted average of causal effects since the strata are not viewed as counter-factual features in such experiments.

\subsection{Saturated (SAT) regression}\label{sec:sat}

We now turn our attention to the last set of estimands we study in this paper; the slope coefficient $\Delta_{\rm sat}$ in \eqref{eq:sat-regression}. As we have stated in the introduction, under Assumption \ref{ass:joint-SOO}, it follows that $\mu(d,a)$ is immediately identified from $E[Y|D=d,A=a]$ for any $d\in \{0,1\}$ and $a\in\mathcal A$ and so identification of any contrast of means $\mu(d,a)$ is straightforward. From this, it immediately follows that the same result could be achieved by running a saturated regression, as in \eqref{eq:sat-regression}, that we re-write here for readability,
$$ Y =  \sum_{a\in\mathcal A} \gamma(a)I\{A=a\}+ \sum_{a\in\mathcal A} \Delta_{\rm sat}(a)I\{A=a\}  D + \epsilon~.$$
Under Assumption \ref{ass:joint-SOO}, standard results on saturated regressions imply that $\Delta_{\rm sat}(a)=\mu(1,a)-\mu(0,a)$ for all $a\in\mathcal A$, and so $\Delta_{\rm sat}(a)$ captures an average causal partial effect of $D$ on $Y$ for each value of the other actions, $a\in\mathcal A$, aligned with Definition \ref{def:PCE}. For completeness, we state and prove this result formally in Theorem \ref{thm:sat} in the appendix. The same result also shows that replacing Assumption \ref{ass:joint-SOO} with Assumption \ref{ass:basic-SOO} leads to a decomposition of $\Delta_{\rm sat}(a)$ that includes a selection term, as it was also the case for the other estimands we considered. Finally, we note that this approach is equivalent to a multi-way ANOVA, since the mean for the outcome variable depends exclusively on categorical variables.

\section{Concluding Remarks}\label{sec:conclusion}

In this paper we analyze settings where the researcher seeks to estimate the average \emph{direct} causal effect of a binary treatment \( D \) on an outcome \( Y \), in situations where the outcome is realized with a delay and other endogenous actions \( A \) may occur in the interim. A central premise of our analysis is that researchers often use regression-based estimands with the intent of isolating a direct causal effect of \( D \) on \( Y \), not a total effect that also includes indirect pathways through \( A \).

We focus on several popular estimands arising from regressions that control for \( A \) in various ways—linearly, with interactions, or not at all—and decompose them to assess whether they admit a causal interpretation. Our decompositions are algebraic and do not rely on identification assumptions, but we examine the additional conditions under which causal interpretations are possible. The results are ``negative" in spirit: they show that, even under strong exogeneity assumptions, these commonly used regression-based estimands generally do not recover the average direct causal effect of interest unless we impose restrictive assumptions (e.g., linearity in the potential outcomes model). We highlight that only in the special case where the actions are binary and mutually exclusive does the long or interacted regression estimand admit a causal interpretation of a direct effect without further assumptions—and our result shows that this is an \emph{if and only if} condition.

The goal of our exercise is not to propose new estimators or to discuss nonparametric identification of direct and indirect effects, but rather to clarify the conditions under which standard approaches yield meaningful causal parameters. We hope this work helps practitioners recognize the limitations of common estimands and adopt a more deliberate approach when interpreting them as capturing direct effects.

\bibliography{ate_ref.bib}
\addcontentsline{toc}{section}{References}

\renewcommand{\theequation}{\Alph{section}-\arabic{equation}}

\setcounter{lemma}{0}
\setcounter{theorem}{0}
\setcounter{corollary}{0}
\setcounter{equation}{0}
\setcounter{remark}{0}

\begin{appendix}

\input{appendix.tex}
\end{appendix}

\end{document}

%% file: appendix.tex

\small 

\section{Appendix}

\subsection{Proofs}\label{app:proofs}

\begin{proof}[Proof of Theorem \ref{thm:short}]
This proof follows from derivations in Section \ref{sec:short} and basic algebraic manipulations. 
\end{proof}


\begin{theorem}\label{thm:long_pre}
Consider the long regression in \eqref{eq:short-regression} and assume that $P\{D=d,A=a\}>0$ for all $(d,a)\in \mathcal D \times \mathcal A$.  Assume that the variance-covariance matrix of $(A,D)$ is positive definite and let $M=\cov(D,A) \var(A) ^{-1}$. Then, 
\begin{align}
    \Delta _{\rm long}
    &=\sum_{a\in \mathcal{A}} \omega^{\rm l}_{\rm dce}(a)E[ Y( 1,a) -Y( 0,a)|D=1,A=a] \notag  \\ 
    &+\sum_{a\in \mathcal{A}}\omega^{\rm l}_{\rm ind}(a)( E[ Y( 0,a) |D=0,A=a] -E[ Y( 0,0) |D=0,A=0] )  \notag \\ 
&+\sum_{a\in \mathcal{A}}\omega^{\rm l}_{\rm dce}(a)( E[ Y( 0,a) |D=1,A=a]-E[ Y( 0,a) |D=0,A=a] )~,\label{eq:long_pre1}
\end{align}
where
\begin{align}
\omega^{\rm l}_{\rm dce}(a) & \equiv  \frac{\pi _{1}( a) [ \var(D) -P\{D=1\} \sum_{j=1}^{K}M_j( a_{j}-E[ A_{j}] ) ] }{\var(D) -\cov(D,A) \var(A) ^{-1} \cov(A,D)} \notag\\
\omega^{\rm l}_{\rm ind}(a) & \equiv  \frac{\var(D) [ \pi _{1}( a) -\pi _{0}( a) ] -P\{A=a\} \sum_{j=1}^{K}M_j( a_{j}-E[ A_{j}] ) }{\var(D) -\cov(D,A) \var(A) ^{-1} \cov(A,D) }~,\label{eq:long_pre2}
\end{align}
and $\pi_{d}(a)$ is defined in \eqref{eq:pis}. Furthermore, $\sum_{a\in \mathcal{A}}\omega^{\rm l}_{\rm dce}(a)=1$, $\sum_{a\in \mathcal{A}}a\omega _{\mathrm{ind}}^{\rm l}(a)={\bf 0}$, and $\sum_{a\in \mathcal{A}}\omega^{\rm l}_{\rm ind}(a)=0$.
\end{theorem}
\begin{proof}
Let $\theta \equiv (\theta _{j}:j=1,\dots ,K)$. By properties of projections, 
\begin{equation}
E[ ( 1,D,A^{\prime })' (Y-( \Delta _{\rm long}D+\theta _{0}+\theta ^{\prime }A)) ]  =  {\bf 0}~. \label{eq:long_pre3}
\end{equation}
Profiling $\theta_{0}$ leads to
\begin{align}
\cov(D,Y)  & = \var(D) \Delta _{\rm long}+\cov(A,D)^{\prime }\theta   \label{eq:long_pre4} \\
\cov(A,Y)  & = \cov(A,D) \Delta _{\mathrm{long} }+\var(A) \theta~.\label{eq:long_pre4B}
\end{align}
Since $\cov(D,A) $ is positive definite, $\var(A) $ is positive definite. Then, \eqref{eq:long_pre4B} implies that $\theta =\var(A) ^{-1}( \cov(A,Y) -\cov(A,D) \Delta _{\rm long})$. If we plug this into \eqref{eq:long_pre4}, we get
\begin{align}
( \var(D) - \cov(D,A) \var(A) ^{-1} \cov(A,D) ) \Delta _{ \mathrm{long}} =  \cov(D,Y) -M\cov(A,Y)~.  \label{eq:long_pre6}    
\end{align}
Since $\cov(D,A)$ is positive definite, $\cov(D,A) \var(A) ^{-1} \cov(A,D) >0$, and so \eqref{eq:long_pre6} implies that
\begin{equation}
\Delta _{\rm long}
  =  \frac{\var(D) \Delta_{\rm short}-\sum_{j=1}^{K}M_j\cov(A_{j},Y) }{\var(D) -\cov(D,A) \var(A) ^{-1} \cov(A,D) }~,
\label{eq:long_pre7}
\end{equation}
where we used that $\var(D) \Delta_{\rm short} = \cov(D,Y)$. For any $j=1,\dots ,K$, some algebra shows that
\begin{align}
\cov&(A_{j},Y)  ~=~ \sum_{a\in \mathcal{A}} E[ Y( 1,a) -Y( 0,a) |D=1,A=a] ( a_{j}-E[ A_{j}] ) \pi _{1}( a) E[D] \notag \\
&+ \sum_{a\in \mathcal{A}} ( E[ Y( 0,a) |D=0,A=a] -E[ Y(0,0) |D=0,A=0] ) ( a_{j}-E[ A_{j}] )P\{A=a\} \notag \\
&+\sum_{a\in \mathcal{A}} ( E[ Y( 0,a) |D=1,A=a] -E[ Y(0,a) |D=0,A=a] ) ( a_{j}-E[ A_{j}] )\pi _{1}( a)E[D]~. \label{eq:long_pre8}
\end{align}
Then, \eqref{eq:long_pre1} follows from combining \eqref{eq:short-decomposition}, \eqref{eq:long_pre7}, and \eqref{eq:long_pre8}. 

To show $\sum_{a\in \mathcal{A}}\omega _{\mathrm{dce}}^{\mathrm{l}}(a)=1$, consider the following derivation.
\begin{align*}
\sum_{a\in \mathcal{A}}\omega _{\mathrm{dce}}^{\mathrm{l}}(a) &\overset{(1)}{=}\frac{ \var(D)-P\{D=1\}\sum_{j=1}^{K}M_{j}(E[A_{j}|D=1]-E[A_{j}])]}{\var(D)-\cov (D,A)\var(A)^{-1}\cov(A,D)} \\
&\overset{(2)}{=}\frac{\var(D)-M\cov(A,D)}{\var(D)-\cov(D,A)\var(A)^{-1}\cov(A,D)} \overset{(3)}{=}1~,
\end{align*}
where (1) holds by $\sum_{a\in \mathcal{A}}\pi _{1}(a)=1$ and $\sum_{a\in \mathcal{A}}\pi _{1}(a)a_{j}=E[A_{j}|D=1]$, and (2) holds by $ P\{D=1\}(E[A_{j}|D=1]-E[A_{j}])=\cov(A_{j},D)$, and (3) holds by definition of $M$.

We show $\sum_{a\in \mathcal{A}}\omega _{\mathrm{ind}}^{\mathrm{l}}(a)=0$ by the following derivation applied to its numerator:
\[
\var(D)\sum_{a\in \mathcal{A}}[\pi _{1}(a)-\pi _{0}(a)]-\sum_{a\in \mathcal{A }}P\{A=a\}\sum_{j=1}^{K}M_{j}(a_{j}-E[A_{j}])=0~,
\]
where the equality holds by $\sum_{a\in \mathcal{A}}\pi _{1}(a)=\sum_{a\in \mathcal{A }}\pi _{0}(a)=\sum_{a\in \mathcal{A}}P\{A=a\}=1$ and $\sum_{a\in \mathcal{A} }P\{A=a\}a_{j}=E[A_{j}]$.

Finally, we show $\sum_{a\in \mathcal{A}}a_{u}\omega _{\mathrm{ind}}^{ \mathrm{l}}(a)=0$ for any $u=1,\ldots ,K$. Once again, we focus on the following derivation applied to its numerator:
\begin{align*}
&\sum_{a\in \mathcal{A}}a_{u}\var(D)[\pi _{1}(a)-\pi _{0}(a)]-\sum_{a\in \mathcal{A}}a_{u}P\{A =a\}\sum_{j=1}^{K}M_{j}(a_{j}-E[A_{j}]) \\
&\overset{(1)}{=}\var(D)[E[ A_{u}|D=1] -E[ A_{u}|D=0] ]-M\cov (A_{u},A) \\
&\overset{(2)}{=}\cov(D,A_{u})-\cov(D,A)\var(A)^{-1}\cov(A,A_{u}) \overset{(3)}{=}0~,
\end{align*}
where (1) holds by $\sum_{a\in \mathcal{A}}\pi _{d}(a)a_{j}=E[A_{j}|D=d]$ for $d=0,1,$ and $\sum_{a\in \mathcal{A}}a_{u}P\{A=a\}(a_{j}-E[A_{j}])=\cov (A_{u},A_{j})$, (2) holds by $\var(D)[E[ A_{u}|D=1] -E[ A_{u}|D=0] ]=\cov(D,A_{u})$ and the definition of $M$, and (3) holds by the fact that $\var(A)^{-1}\cov(A,A_{u})$ equals a column vector with zeros except for a one in the $u$th position.
\end{proof}

\sloppy
\begin{proof}[Proof of Theorem \ref{thm:long}]
The first part follows from Theorem \ref{thm:long_pre}, which also yields $\sum_{a\in \mathcal{A}}\omega^{\rm l}_{\rm dce}(a)=1$ and $\sum_{a\in \mathcal{A}}\omega^{\rm l}_{\rm ind}(a)=0$. To complete the proof, we now show the equivalence between (a), (b), (c), and (d).

First, we show that (a) implies (b), (c), and (d). To this end, assume (a) holds. Then, the long regression in \eqref{eq:long-regression} is equivalent to an SFE regression in \eqref{eq:sfe-regression}. To see why, note that (a) implies that $A$ is a $K$-dimensional vector that is either equal to 0 or a canonical vector (i.e., a vector with a 1 in only one of its coordinates and zeroes otherwise). If we then let $\theta (a) =\theta _{0}$ for $a=0$ and $\theta (a) =\theta_{j}$ for $a$ being the canonical vector with $j$th coordinate equal to one, we get
\begin{equation*}
\theta _{0}+\theta ^{\prime }A = \sum_{a\in \mathcal{A}}\theta (a) I\{ A=a\}~.
\end{equation*}
Therefore, $\Delta _{\rm long}=\Delta _{\mathrm{sfe}}$ and Theorem \ref{thm:sfe_pre} imply (b) (with $\omega^{\rm l}_{\rm dce}(a)=\omega_{\rm sfe}(a)$), (c), and (d).

Second, we show that (b), (c), or (d) implies (a) or, equivalently, the negation of (a) implies the negation of (b), the negation of (c), and the negation of (d).

Start by considering the case when $K=1$. Then, (a) fails when $\mathcal{A}\neq\{ 0,1\} $. 
For example, consider the case where $\mathcal{A}=\{ 0,1,2\} $ with $\{ A|D=1\} \sim {\rm Bi}( 2,0.3) $, $\{ A|D=0\} \sim {\rm Bi}( 2,0.9) $, and $ P\{D=1\} =0.5$, where ${\rm Bi}(n,p)$ denotes a Binomial distribution with $n$ trials and probability $p$. With this distribution of $(A,D)$, the weights in  \eqref{eq:long_pre2} become $\omega^{\rm l}_{\rm dce}\approx [ -0.1,0.76,0.34] $ and $\omega^{\rm l}_{\rm ind}\approx [ -0.14,0.28,-0.14] $, and so (b) and (c) fail.
To verify that (d) fails, consider the case with $\mu(0,a) = 0$ for all $a \in \mathcal{A}$, $\mu(1,0) = 12$, and $ \mu(1,1) =\mu(1,2) = 1$. We then have $\mu(1, a) - \mu(0, a) > 0$ for all $a \in \mathcal{A}$, yet $\Delta_{\mathrm{long}} = \Delta_{\mathrm{dce}}^{l} = -0.1<0$, and so (d) fails.

Next, consider the case when $K=2$. In this case, (a) can fail when (i) $\mathcal{ A}_{j}\not=\{ 0,1\} $ for some $j=1,2$ or (ii) $\mathcal{A}_{j}=\{ 0,1\} $ for all $j=1,2$ but $A_{1}A_{2}\not=0$. 
For (i), consider $\{ A_{1}|D=0\} \sim {\rm Bi}(2,0.9) $, $\{ A_{1}|D=1\} \sim {\rm Bi}(2,0.3) $, $P\{D=1\} =0.5$, $A_{2}\perp \{ D,A_{1}\} $, and $\var(A_{2}) >0$. The fact that $A_{2}\perp \{ D,A_{1}\} $ and $\var(A_{2}) >0$ implies that $A_{2}$ drops out of the expressions in \eqref{eq:long_pre2}, and the example becomes identical to the one considered when $K=1$, where (b), (c), and (d) fail. 
For (ii), let ${\rm Ber}(p)$ denote a Bernoulli distribution with parameter $p$ and consider $\{ A_{j}|D=0\} \sim {\rm Ber}( 0.1) $ and $ \{ A_{j}|D=1\} \sim {\rm Ber}( 0.7) $ for $j=1,2$, with  $P\{D=1\} =0.5$, so that $\mathcal{A}_{j}=\{ 0,1\} $ for $j=1,2$ and $P\{A_{1}A_{2}=0\} \approx 0.45$. With this distribution of $(A,D)$, the weights in \eqref{eq:long_pre2} become $\omega^{\rm l}_{\rm dce}\approx [ 0.34,0.38,0.48,-0.10] $ and $\omega^{\rm l}_{\rm ind}\approx [ -0.14,0.14,0.14,-0.14] $, and so (b) and (c) fail. To verify that (d) fails, consider the case with $\mu(0,a) = 0$ for all $a \in \mathcal{A}$, $\mu(1,(0,0)) =\mu(1,(1,0))=\mu(1,(0,1)) = 1$, and $\mu(1,(1,1)) = 13$. We then have $\mu(1, a) - \mu(0, a) > 0$ for all $a \in \mathcal{A}$, yet $\Delta_{\mathrm{long}} = \Delta_{\mathrm{dce}}^{l} = -0.1<0$, and so (d) fails.

Finally, consider the case $K>2$. Then, (a) can fail when (i) $\mathcal{A}_{j}\not=\{ 0,1\} $ for some $j=1,\dots ,K$ or (ii) $\mathcal{A}_{j}=\{ 0,1\} $ for all $j=1,\dots ,K$ but $A_{j}A_{l}\not=0$ for some $j,l=1,\dots ,K$ with $j\not=l$. In either case, we can repeat the examples used for $K=2$ by adding coordinates $j=3,\dots ,K$ with $\{ A_{j}:j>2\} \perp \{ D,\{ A_{j}:j\leq 2\} \}$, and $\var(A_{j}) >0$ for $j>2$. By construction, $\{ A_{j}:j>2\} $ drops out of the expressions in \eqref{eq:long_pre2}, and the examples considered with $K=2$ imply the failure of (b), (c), and (d).
\end{proof}

\sloppy
\begin{theorem}\label{thm:inter_pre} 
Consider the long regression with interactions in \eqref{eq:inter-regression} and assume that $P\{D=d,A=a\}>0$ for all $(d,a)\in \mathcal D\times \mathcal A$. Assume that the variance-covariance matrix of $(A,D,AD)$, denoted y $\Sigma_{\rm inter}$, is positive definite and let $M=\cov(D,W) \var(W)^{-1}$ with $W\equiv (A',A'D)'$. Then,  
\begin{align}
\Delta _{\mathrm{inter}} &= \sum_{a\in \mathcal{A}} \omega^{\rm i}_{\rm dce}(a)E[Y(1,a)-Y(0,a)|D=1,A=a] \notag \\ 
&+\sum_{a\in \mathcal{A}}\omega^{\rm i}_{\rm ind}(a)(E[Y(0,a)|D=0,A=a]-E[Y(0,0)|D=0,A=0]) \notag \\ 
&+\sum_{a\in \mathcal{A}}\omega^{\rm i}_{\rm dce}(a)(E[Y(0,a)|D=1,A=a]-E[Y(0,a)|D=0,A=a]) \label{eq:inter_pre1}~,  
\end{align}
where 

{\footnotesize 
\begin{align}
\omega^{\rm i}_{\rm dce}(a)& \equiv \frac{\pi _{1}(a)\left[
\sigma^2_{D}-p\sum_{j=1}^{K}M_{j}(a_{j}-E[A_{j}])-p\sum_{j=1}^{K}M_{j+K}( a_{j}-p E[ A_{j}|D=1] )\right]}{\sigma^2_{D}-M\cov(W,D)} \label{eq:inter_pre2} \\
\omega^{\rm i}_{\rm ind}(a)& \equiv  \frac{\sigma^2_{D}(\pi _{1}(a)-\pi_{0}(a)) -\sum_{j=1}^{K}M_{j}p_{a}(a_{j}-E[A_{j}])-p\sum_{j=1}^{K}M_{j+K}( \pi _{1}(a)a_{j}-p_a E[ A_{j}|D=1] )}{\sigma^2_{D}-M\cov(W,D)}~,\notag 
\end{align}
}
$p=P\{D=1\}$, $p_{a} = P\{A=a\}$, $\sigma^2_{D}=\var(D)$, and $\pi_{d}(a)$ is defined in \eqref{eq:pis}. Furthermore, $\sum_{a\in \mathcal{A}}\omega^{\rm i}_{\rm dce}(a)=1$, $\sum_{a\in \mathcal{A}}a\omega^{\rm i}_{\rm ind}(a)={\bf 0}$, and $\sum_{a\in \mathcal{A}}\omega^{\rm i}_{\rm ind}(a)=0$.
\end{theorem}
\begin{proof}
Let $\theta =(\theta _{j}:j=1,\dots ,K)$, $\lambda = (\lambda _{j}:j=1,\dots ,K)$, and $\alpha = (\theta',\lambda')'$. By properties of projections, 
\begin{equation}
E[( 1,D,A^{\prime },DA^{\prime })'(Y-(\Delta _{\mathrm{inter}}D+\theta _{0}+\alpha ^{\prime }W))]~=~{\bf 0}~.  \label{eq:inter_pre3}
\end{equation}
Profiling $\theta_0$ leads to, 
\begin{align}
\cov(D,Y)&  = \var(D)\Delta _{\mathrm{inter}}+\cov(W,D)^{\prime }\alpha \label{eq:inter_pre4} \\
\cov(W,Y)&  = \cov(W,D)\Delta _{\mathrm{inter}}+\var(W)\alpha ~.\label{eq:inter_pre4B}
\end{align}
Since $\Sigma_{\mathrm{inter}} $ is positive definite, $\var(W)$ is positive definite. Then, \eqref{eq:inter_pre4B} implies that $\alpha =\var(W)^{-1}(\cov(W,Y)-\cov(W,D)\Delta _{\mathrm{inter}})$. If we plug this into \eqref{eq:inter_pre4}, we get 
\begin{align}
(\var(D)-M\cov(W,D))\Delta _{\mathrm{ inter}}
&=\cov(D,Y)-M\cov(W,Y)~.\label{eq:inter_pre6}
\end{align}
Since $\Sigma_{\mathrm{inter}} $ is positive definite, $\var(D)-\cov(W,D)^{\prime }\var(W)^{-1}\cov(W,D)>0$ and so \eqref{eq:inter_pre6} implies that 
\begin{equation}
\Delta _{\mathrm{inter}} = \frac{\cov(D,Y)-\sum_{j=1}^{K}M_{j}\cov(A_{j},Y)-\sum_{j=1}^{K}M_{j+K}\cov(DA_{j},Y)}{\var(D)-M\cov(W,D)}~,\label{eq:inter_pre7}
\end{equation}
where we used that $\var(D)\Delta_{\rm short}=\cov(D,Y)$. For any $j=1,\dots ,K$, some algebra shows that
\begin{align}
\cov(&A_{j},Y) = \sum_{a\in \mathcal{A}}E[Y(1,a)-Y(0,a)|D=1,A=a](a_{j}-E[A_{j}])\pi _{1}(a)p \notag\\ 
&+ \sum_{a\in \mathcal{A}}(E[Y(0,a)|D=0,A=a]-E[Y(0,0)|D=0,A=0])(a_{j}-E[A_{j}])p_{a}\notag\\ 
&+ \sum_{a\in \mathcal{A}}(E[Y(0,a)|D=1,A=a]-E[Y(0,a)|D=0,A=a])(a_{j}-E[A_{j}])\pi _{1}(a)p~,\label{eq:inter_pre8}
\end{align}
and
\begin{align}
&\cov(DA_{j},Y)= p \sum_{a\in \mathcal{A}}E[ Y( 1,a) -Y( 0,a) |D=1,A=a] \pi_{1}(a)( a_{j}-p E[ A_{j}|D=1] )\notag\\ 
&+ p\sum_{a\in \mathcal{A}}( E[ Y( 0,a) |D=1,A=a] -E[ Y(0,a) |D=0,A=a] ) \pi _{1}(a)( a_{j}-p E[ A_{j}|D=1] )  \notag\\ 
&+p\sum_{a\in \mathcal{A}}( E[ Y( 0,a) |D=0,A=a] -E[Y(0,0)|D=0,A=0])( \pi _{1}(a)a_{j}-p_a E[ A_{j}|D=1] ) ~.\label{eq:inter_pre9}
\end{align}
By plugging in \eqref{eq:short-decomposition}, \eqref{eq:inter_pre8}, and \eqref{eq:inter_pre9} into \eqref{eq:inter_pre7}, \eqref{eq:inter_pre1} follows. 

To show $\sum_{a\in \mathcal{A}}\omega _{\mathrm{dce}}^{\mathrm{i}}(a)=1$, consider the following derivation.
\begin{align*}
\sum_{a\in \mathcal{A}}\omega _{\mathrm{dce}}^{\mathrm{i}}(a)
&\overset{(1)}{=}\frac{ \sigma _{D}^{2}-p\sum_{j=1}^{K}M_{j}(E[A_{j}|D=1]-E[A_{j}])-\sum_{j=1}^{K}M_{j+K} \sigma _{D}^{2}E[A_{j}|D=1]}{\sigma _{D}^{2}-M\cov(W,D)} \\
&\overset{(2)}{=}\frac{\sigma _{D}^{2}-\sum_{j=1}^{K}M_{j}\cov(D,A_{j})- \sum_{j=1}^{K}M_{j+K}\cov(D,DA_{j})}{\sigma _{D}^{2}-M\cov(W,D)}\overset{(3)}{=}1~,
\end{align*}
where (1) holds by $\sum_{a\in \mathcal{A}}\pi _{1}(a)=1$, $\sum_{a\in \mathcal{A}}\pi _{1}(a)a_{j}=E[A_{j}|D=1]$, and $\sigma _{D}^{2}=p( 1-p) $, (2) holds by $p(E[A_{j}|D=1]-E[A_{j}])=\cov(D,A_{j})$ and $ \sigma _{D}^{2}E[A_{j}|D=1]=\cov(DA_{j},D)$, and (3) holds by definition of $ M$.

We show $\sum_{a\in \mathcal{A}}\omega _{\mathrm{ind}}^{\mathrm{i}}(a)=0$ by the following derivation applied to its numerator:
\begin{equation*}
\sigma _{D}^{2}\sum_{a\in \mathcal{A}}(\pi _{1}(a)-\pi _{0}(a))-\sum_{j=1}^{K}M_{j}\sum_{a\in \mathcal{A}}p_{a}(a_{j}-E[A_{j}])-p \sum_{j=1}^{K}M_{j+K}\sum_{a\in \mathcal{A}}(\pi _{1}(a)a_{j}-p_{a}E[A_{j}|D=1])=0~,
\end{equation*}
where the equality holds by $\sum_{a\in \mathcal{A}}\pi _{1}(a)=\sum_{a\in \mathcal{A }}\pi _{0}(a)=\sum_{a\in \mathcal{A}}p_{a}=1$, $\sum_{a\in \mathcal{A} }p_{a}a_{j}=E[A_{j}]$, and $\sum_{a\in \mathcal{A}}\pi _{1}(a)a_{j}=E[A_{j}|D=1]$.

Finally, we show $\sum_{a\in \mathcal{A}}a_{u}\omega _{\mathrm{ind}}^{ \mathrm{i}}(a)=0$ for any $u=1,\ldots ,K$. Once again, we focus on the following derivation applied to its numerator:
{\footnotesize
\begin{align*}
\sum_{a\in \mathcal{A}}a_{u}\sigma _{D}^{2}(\pi _{1}(a)-\pi _{0}(a))&-\sum_{j=1}^{K}M_{j}\sum_{a\in \mathcal{A} }a_{u}p_{a}(a_{j}-E[A_{j}])-p\sum_{j=1}^{K}M_{j+K}\sum_{a\in \mathcal{A} }a_{u}(\pi _{1}(a)a_{j}-p_{a}E[A_{j}|D=1]) \\
&\overset{(1)}{=}\sigma _{D}^{2}[E[ A_{u}|D=1] -E[ A_{u}|D=0] ]\\
&-\sum_{j=1}^{K}M_{j}\cov(A_{u},A_{j})- \sum_{j=1}^{K}M_{j+K}p(E[A_{u}A_{j}|D=1]-E[A_{u}]E[A_{j}|D=1]) \\
&\overset{(2)}{=}\cov(D,A_{u})-\sum_{j=1}^{K}M_{j}\cov(A_{j},A_{u})- \sum_{j=1}^{K}M_{j+K}\cov(DA_{j},A_{u}) \\
&\overset{(3)}{=}\cov(D,A_{u})-\cov(D,W)\var(W)^{-1}cov( W,A_{u}) \overset{(4)}{=}0~,
\end{align*}
}
where (1) holds by $\sum_{a\in \mathcal{A}}\pi _{d}(a)a_{j}=E[A_{j}|D=d]$ for $d=0,1,$ $\sum_{a\in \mathcal{A}}a_{u}P\{A=a\}(a_{j}-E[A_{j}])=\cov (A_{u},A_{j})$, and $\sum_{a\in \mathcal{A}}a_{u}p_{a}=E[ A_{u}] $, (2) holds by $\var(D)[E[ A_{u}|D=1] -E[ A_{u}|D=0] ]= \cov(D,A_{u})$ and $p(E[A_{u}A_{j}|D=1]-E[A_{u}]E[A_{j}|D=1])=\cov (DA_{j},A_{u})$, (3) holds by the definition of $M$, and (4) holds by the fact that $\var(W)^{-1}\cov( W,A_{u}) $ equals a column vector with zeros except for a one in the $u$th position.
\end{proof}

\begin{proof}[Proof of Theorem \ref{thm:inter}]
The first part follows from Theorem \ref{thm:inter_pre}, which also yields that $\sum_{a\in \mathcal{A}}\omega^{\rm i}_{\rm dce}(a)=1$ and $\sum_{a\in \mathcal{A}}\omega^{\rm i}_{\rm ind}(a)=0$. To complete the proof, we now show the equivalence between (a), (b), (c), and (d).

First, we show that (a) implies (b), (c), and (d). To this end, assume (a) holds. Then, the long with interactions regression in \eqref{eq:inter-regression} is equivalent to an SAT regression in \eqref{eq:sfe-regression}. To see why, note that (a) implies that $\mathcal{A}=\{{\bf 0}_{K\times 1},\{e_{j}:j=1,\dots ,K\}\}$, where $e_{j}\in \mathbb{R}^{K\times 1}$ has a one in the $j$'th coordinate and zero otherwise. By defining $A_{0}=1-\sum_{j=1}^{K}A_{j}$, $\gamma (a)=\theta _{0}$ and $\Delta _{\mathrm{sat}}(a)=\Delta _{\mathrm{inter }}$ for $a=0$, and $\gamma (a)=\theta _{0}+\theta _{j}$ and $\Delta _{ \mathrm{sat}}(a)=\Delta _{\mathrm{inter}}+\lambda _{j}$ for $a=e_{j}$ with $j=1,\dots ,K$, we get 
\begin{equation*}
\Delta _{\mathrm{inter}}D+\theta _{0}+\theta ^{\prime }A+\lambda ^{\prime }AD ~= ~\sum_{a\in \mathcal{A}}\gamma (a)I\{A=a\}+\sum_{a\in \mathcal{A}}\Delta _{\mathrm{sat}}(a)I\{A=a\}D~.
\end{equation*}
Therefore, $\Delta _{\mathrm{inter}}=\Delta _{\mathrm{sat}}(0)$ and Theorem \ref{thm:sat} imply (b) (with $\omega^{\rm i}_{\rm dce}(0)=1$ and $\omega^{\rm i}_{\rm dce}(e_{j})=0$ for $j=1,\dots ,K$), (c), and (d). 

To conclude, we now show that (b), (c), or (d) implies (a) or, equivalently, the negation of (a) implies the negation of (b), the negation of (c), and the negation of (d).

First, consider the case when $K=1$. Then, (a) fails when $\mathcal{A}\neq \{0,1\}$. 
For example, if $\{ A|D=0\} \sim {\rm Bi}(2,0.3)$, $\{ A|D=1\} \sim {\rm Bi}(2,0.9)$, and $P\{D=1\} =0.5$, and so $ \mathcal{A}=\{0,1,2\}$. By evaluating this information on \eqref{eq:inter_pre2}, we get $\omega^{\rm i}_{\rm dce}\approx [0.19,1.62,-0.81]$ and $\omega^{\rm i}_{\rm ind}\approx [-0.72,1.44,-0.72]$, i.e., (b) and (c) fail. To verify that (d) fails, consider the case with $\mu(0,a) = 0$ for all $a \in \mathcal{A}$, $\mu(1,0) = \mu(1,1) = 1$, and $\mu(1,2) = 3$. We then have $\mu(1, a) - \mu(0, a) > 0$ for all $a \in \mathcal{A}$, yet $\Delta_{\mathrm{inter}} = \Delta_{\mathrm{dce}}^{i} = -0.62<0$, and so (d) fails.

Second, consider the case when $K=2$. Then, (a) can fail when (i) $\mathcal{A }_{j}\not=\{0,1\}$ for some $j=1,2$ or (ii) $\mathcal{A}_{j}=\{0,1\}$ for all $j=1,2$ but $A_{1}A_{2}\not=0$. 
For (i), consider $\{A_{1}|D=0\}\sim {\rm Bi}(2,0.3)$, $\{A_{1}|D=1\}\sim {\rm Bi}(2,0.9)$, $P\{D=1\} =0.5$, $A_{2}\perp \{D,A_{1}\}$, and $\var(A_{2})>0$. The fact that $A_{2}\perp \{D,A_{1}\}$ and $ \var(A_{2})>0$ implies that $A_{2}$ drops out of the expressions in \eqref{eq:inter_pre2}, and the example becomes identical to the one considered when $K=1$ and, thus, (b), (c), and (d) fail. 
For (ii), consider $ \{A_{j}|D=0\}\sim {\rm Be}(0.3)$ and $\{A_{j}|D=1\}\sim {\rm Be}(0.9)$ for $j=1,2$, and $ P(D=1)=0.5$, and so $\mathcal{A}_{j}=\{0,1\}$ for $j=1,2$ and $P(A_{1}A_{2}=0)\approx 0.25$. By evaluating this information on \eqref{eq:inter_pre2}, we get $\omega^{\rm i}_{\rm dce}\approx [0.19,0.81,0.81,-0.81]$ and $\omega^{\rm i}_{\rm ind}\approx [ -0.72,0.72,0.72,-0.72]$, i.e., (b) and (c) fail. To verify that (d) fails, consider the case with $\mu(0,a) = 0$ for all $a \in \mathcal{A}$, $\mu(1,(0,0)) = \mu(1,(1,0)) = \mu(1,(0,1)) = 1$, and $\mu(1,(1,1)) = 3$. We then have $\mu(1, a) - \mu(0, a) > 0$ for all $a \in \mathcal{A}$, yet $\Delta_{\mathrm{inter}} = \Delta_{\mathrm{dce}}^{i} = -0.62<0$, and so (d) fails.

Finally, consider $K>2$. Then, (a) can fail when (i) $\mathcal{A}_{j}\not=\{ 0,1\} $ for some $j=1,\dots ,K$ or (ii) $\mathcal{A}_{j}=\{ 0,1\} $ for all $j=1,\dots ,K$ but $A_{j}A_{l}\not=0$ for some $j,l=1,\dots ,K$ with $j\not=l$. In either case, we can repeat the examples used for $K=2$ by adding coordinates $j=3,\dots ,K$ with $\{ A_{j}:j>2\} \perp \{ D,\{A_{j}:j\leq 2\} \}$, and $\var(A_{j}) >0$ for $j>2$. By construction, $\{ A_{j}:j>2\} $ drops out of the expressions in \eqref{eq:inter_pre2}, and the examples considered with $K=2$ imply the failure of (b), (c), and (d).
\end{proof}

\begin{theorem}\label{thm:sfe_pre} 
Consider the SFE regression in \eqref{eq:sfe-regression}, and assume that $P\{D=d,A=a\}>0$ for all $(d,a)\in \mathcal D\times \mathcal A$. Then, 
\begin{equation}\label{eq:sfe-decomposition-soo2}
  \Delta_{\rm sfe}= \Delta_{\rm dce}^{\rm f}+ \Delta_{\rm sel}^{\rm f}~, 
\end{equation}
where
\begin{align}
\omega_{\mathrm{sfe}}(a)&~\equiv ~\frac{P\{D=0|A=a\}  P\{ D=1|A=a\}  P\{ A=a\}  }{\sum_{\tilde{a}\in \mathcal{A}}P\{ D=1|A=\tilde{a}\}  P\{ D=0|A=\tilde{a}\}  P\{ A=\tilde{a}\} }~~~\text{ for all }a \in \mathcal{A} \label{eq:sfe_2}\\
  \Delta_{\rm dce}^{\rm f} & ~\equiv ~ \sum_{a\in\mathcal A}  \omega_{\rm sfe}(a) E[Y(1,a) - Y(0,a)|D=1,A=a] \label{eq:sfe-dce2} \\
  \Delta_{\rm sel}^{\rm f} & ~\equiv ~ \sum_{a\in\mathcal A}  \omega_{\rm sfe}(a) ( E[Y(0,a)|D=1,A=a]-E[Y(0,a)|D=0,A=a] )\label{eq:sfe-sel2}~.
\end{align}
Furthermore, note that $\sum_{a\in \mathcal{A}}\omega_{\mathrm{sfe}}(a)=1$ and $\omega_{\mathrm{sfe}}(a)\geq 0$.
\end{theorem}
\begin{proof}
By properties of projections, 
\begin{align}
E[ YD]  & = \Delta _{\mathrm{sfe}}E[ D] +\sum_{a\in \mathcal{A}}\theta ( a) E[ I\{A=a\} D] \label{eq:sfe_4} \\
E[ Y I\{ A=a\} ]  & = \Delta _{\mathrm{sfe}}E[ DI\{A=a\} ] +\theta ( a) P\{A=a\} \text{ for all }a\in \mathcal{A}~.\label{eq:sfe_5}
\end{align}
By $P\{A=a\}>0$ for all $a\in \mathcal{A}$, \eqref{eq:sfe_5} implies that
\begin{equation}
\theta ( a)  = E[ Y|A=a] -\Delta _{\mathrm{sfe}}E[ D|A=a]~~~ \text{ for all }a\in \mathcal{A}~.
\label{eq:sfe_6}
\end{equation}
Then, \eqref{eq:sfe_4}, \eqref{eq:sfe_6}, and some algebra imply that
\begin{align}
&E[ Y|D=1] -\sum_{a\in \mathcal{A}}E[ Y|A=a] P\{A=a|D=1\} \notag \\
&=~\Delta _{\mathrm{sfe}} \sum_{a\in \mathcal{A}}\frac{P\{D=1|A=a\}P\{ D=0|A=a\} P\{ A=a\} }{P\{ D=1\}}~,
\label{eq:sfe_7}
\end{align}
Under $P\{ A=a\} >0$ and $P\{ D=1|A=a\} \in ( 0,1) $ for all $a\in \mathcal{A}$, \eqref{eq:sfe_7} implies that  
\begin{align}
\Delta _{\mathrm{sfe}} & ~=~ \frac{P\{ D=1\} E[ Y|D=1]-\sum_{a\in \mathcal{A}}E[ Y|A=a] P\{ A=a,D=1\} }{\sum_{a\in \mathcal{A}}P\{ D=1|A=a\} P\{ D=0|A=a\} P\{A=a\} } \notag\\
& ~=~ \sum_{a\in \mathcal{A}}\omega_{\rm sfe}(a) \big( E[ Y|A=a,D=1] -E[Y|A=a,D=0]  \big)~.\label{eq:sfe_8}
\end{align}
By doing algebra on \eqref{eq:sfe_8}, \eqref{eq:sfe-decomposition-soo2} follows. Finally, verifying $\sum_{a\in \mathcal{A}}\omega_{\mathrm{sfe}}(a)=1$ and $\omega_{\mathrm{sfe}}(a)\geq 0$ is straightforward given the definition in \eqref{eq:sfe_2}.
\end{proof}

\begin{proof}[Proof of Theorem \ref{thm:sfe}]
This result follows immediately from Theorem \ref{thm:sfe_pre}. 
\end{proof}

\begin{theorem}\label{thm:sat} 
Consider the SAT regression in \eqref{eq:sat-regression} and assume that $P\{D=d,A=a\}>0$ for all $(d,a)\in \mathcal D\times \mathcal A$. Then, for all $a\in \mathcal{A}$,
\begin{equation}\label{eq:sat-decomposition-soo}
  \Delta_{\rm sat}(a) =  \Delta_{\rm dce}^{\rm t}(a)+ \Delta_{\rm sel}^{\rm t}(a)~, 
\end{equation}
where
\begin{align}
  \Delta_{\rm dce}^{\rm t}(a) &  \equiv   E[Y(1,a) - Y(0,a)|D=1,A=a] \label{eq:sat-dce} \\
  \Delta_{\rm sel}^{\rm t}(a) &  \equiv   E[Y(0,a)|D=1,A=a]-E[Y(0,a)|D=0,A=a] \label{eq:sat-sel}~.
\end{align}
Furthermore, under Assumption \ref{ass:joint-SOO}, $\Delta_{\rm sel}^{\rm t}(a)=0$ and 
\begin{equation}\label{eq:sat_2}
\Delta_{\mathrm{sat}}( a)  = \Delta_{\rm dce}^{\rm t}(a) = \mu ( 1,a) -\mu (0,a)~.
\end{equation}
\end{theorem}
\begin{proof}
Fix $a\in \mathcal{A}$ arbitrarily throughout this proof. By projection,
\begin{align}
E[ YI\{A=a\}]  & ~=~ \gamma (a)P\{A=a\}+\Delta _{\mathrm{sat}}(a)E[DI\{A=a\}]  \notag \\
E[ YDI\{A=a\}]  & ~=~ (\gamma (a)+\Delta _{\mathrm{sat}}(a))E[DI\{A=a\}]~.\label{eq:sat_3}
\end{align}
By $P\{ A=a\} >0$, \eqref{eq:sat_3} implies that 
\begin{align}
\gamma (a) & ~=~ E[ Y|A=a] -\Delta _{\mathrm{sat}}(a)P\{D=1|A=a\}   \label{eq:sat_4} \\
E[ YD|A=a]  & ~=~ (\gamma (a)+\Delta _{\mathrm{sat}}(a))P\{D=1|A=a\} ~.\label{eq:sat_5}
\end{align}
By plugging in \eqref{eq:sat_4} into \eqref{eq:sat_5}, we get
\begin{equation}
E[ YD|A=a] -E[ Y|A=a] P\{D=1|A=a\}  ~=~ \Delta _{\mathrm{sat}}(a)P\{ D=1|A=a\}P\{ D=0|A=a\} ~.\label{eq:sat_6}
\end{equation}
By \eqref{eq:sat_6} and $P\{ D=1|A=a\} \in ( 0,1) $, we get that
\begin{equation}
\Delta _{\mathrm{sat}}(a) ~=~ E[ Y(1,a)|D=1,A=a] -E[ Y(0,a)|D=0.A=a].\label{eq:sat_7}
\end{equation}
The desired result follows from adding and subtracting $E[ Y(0,a)|D=1,A=1]$ to \eqref{eq:sat_7}.
\end{proof}

\subsection{Auxiliary Lemmas}\label{app:aux-lemmas}

\begin{lemma}\label{lem:seq-igno}
The following statements are true. 
\begin{enumerate}[(a)]
    \item Assumption \ref{ass:seq-igno} implies Assumption \ref{ass:joint-SOO}.
    \item Assumption \ref{ass:joint-SOO} does not imply Assumption \ref{ass:seq-igno}.
    \item Assumption \ref{ass:seq-igno} implies that $ Y(\tilde{d},a)\perp A(d) \mid X$ for $(\tilde{d},d,a)\in  \mathcal D \times \mathcal D\times \mathcal A $.
\end{enumerate}
\end{lemma}
\begin{proof}
    \underline{Part (a)}. For any $( \tilde{d},\tilde{a},y,a,d) $, we have
\begin{align}
P\{ Y(\tilde{d},\tilde{a})\leq y,A( d)=a,D=d|X\}
&~\overset{(1)}{=}~P\{ Y(\tilde{d},\tilde{a})\leq y|X\}~ P\{ A(d)=a|X\}~P\{ D=d|X\}  \notag \\
&~\overset{(2)}{=}~P\{ Y(\tilde{d},\tilde{a})\leq y|X\}~ P\{A( d) =a,D=d|X\}   \notag \\
&~\overset{(3)}{=}~P\{ Y(\tilde{d},\tilde{a})\leq y|X\}~ P\{A=a,D=d|X\}~, \label{eq:ass_4}
\end{align}
where (1) holds by Assumption \ref{ass:seq-igno}, (2) holds by Assumption \ref{ass:seq-igno}(i), and (3) holds by $A( D) =A$. Since $( \tilde{d},\tilde{a},y,a,d)$ is arbitrary, \eqref{eq:ass_4} implies Assumption \ref{ass:joint-SOO}.

   \underline{Part (b)}. Consider the following example. Assume $X\perp ( D,( A( d) :d\in \mathcal{D})' ,( Y(\tilde{d},a):( \tilde{d},a) \in \mathcal{D}\times \mathcal{A})' )' $, $Y(d,a)=0$ for all $(d,a) $, $( A( 1) ,A( 0) ) =(D,D) $, and $D\sim {\rm Be}( 0.5) $. Since $Y(d,a)=0$, it is independent of $( D,A( D) ) =( D,D) $. Thus, Assumption \ref{ass:joint-SOO} holds. By $Y(d,a)=0$ for all $( d,a) $ and also $A( d) =D$, we have $Y(\tilde{d},a)\perp A(d)|D$, so Assumption \ref{ass:seq-igno}(ii) holds. However, $(Y(\tilde{d},a),A(d))=( 0,D) \not\perp D$, and so Assumption \ref{ass:seq-igno}(i) fails.

   \underline{Part (c)}. For any  $( \tilde{d},\tilde{a},y,a,d) $, we have
\begin{align}
P\{ Y(\tilde{d},\tilde{a})\leq y,A(d)=a|X\}  &~\overset{(1)}{=}~P\{Y(\tilde{d},\tilde{a})\leq y,A(d)=a|X,D\}  \notag\\
&~\overset{(2)}{=}~P\{ Y(\tilde{d},\tilde{a})\leq y|X,D\} ~ P\{ A(d)=a|X,D\}  \notag\\
&~\overset{(3)}{=}~P\{ Y(\tilde{d},\tilde{a})\leq y|X\} ~ P\{ A(d)=a|X\} ~,\label{eq:ass_5}
\end{align}
where (1) and (3) hold by Assumption \ref{ass:seq-igno}(i), and (2) holds by Assumption \ref{ass:seq-igno}(ii).  Since $( \tilde{d},\tilde{a},y,a,d) $ is arbitrary, \eqref{eq:ass_5} implies Assumption \ref{ass:joint-SOO}.  
\end{proof}

\begin{lemma}\label{lem:linear-mu}
Assume the conditions in Theorem \ref{thm:long}, and that
\begin{equation}
\mu (d,a)~=~\kappa _{0}+\kappa _{1}d+\kappa _{2}^{\prime }a~~\text{ for all }~(d,a)\in \{0,1\}\times \mathcal{A} \label{eq:imai1}
\end{equation}
for some constants $\kappa_0,\kappa_1,\kappa_2$. First, the coefficients in \eqref{eq:long-regression} satisfy $\Delta_{\mathrm{long}} =\kappa _{1}$, $\theta_0 = \kappa_0$, and $\theta_1 = \kappa_2$. Second, the terms in the decomposition in \eqref{eq:long-decomposition} are $\Delta _{\mathrm{dce}}^{\mathrm{l}}=\kappa _{1}$ and $\Delta _{\mathrm{ind}}^{\mathrm{l}}=0$.
\end{lemma}
\begin{proof}
Assumption \ref{ass:joint-SOO} implies that $E(Y|D=d,A=a)=\mu (d,a)$ which, combined with \eqref{eq:imai1}, implies that the conditional expectation of $Y$ is linear in $(1,a,d)$. From here, the first result holds because the linear regression consistently estimates the parameters of a linear conditional expectation. The second part follows immediately from combining \eqref{eq:imai1} with $\sum_{a\in \mathcal{A}}a\omega _{\mathrm{ind}}^{\mathrm{l}}(a)={\bf 0}$ and $\sum_{a\in \mathcal{A}}\omega _{\mathrm{dce}}^{\mathrm{l}}(a)=1$ (both shown in Theorem \ref{thm:long_pre}).
\end{proof}

\begin{lemma}
The examples used in the proofs of Theorem \ref{thm:long} and \ref{thm:inter} can be completed to satisfy Assumption \ref{ass:seq-igno}.
\end{lemma}
\begin{proof}
For brevity, we focus on the example in the proof of Theorem \ref{thm:long} when $K=1$. A similar argument can be made for all other examples.

Recall that the example in the proof of Theorem \ref{thm:long} when $K=1$ is as follows: $\{ A|D=0\} \sim {\rm Bi}(2,0.3) $, $\{ A|D=1\} \sim {\rm Bi}(2,0.9) $, and $ P\{D=1\} =0.5$, and so $\mathcal{A}=\{ 0,1,2\} $. The example is silent about $X$ or $\{Y(d,a):( {d},a) \in \mathcal{D}\times \mathcal{A}\}$, and so it is unclear whether Assumption \ref{ass:seq-igno} holds or not. We now provide one way to complete the specification of the example in a manner compatible with Assumption \ref{ass:seq-igno}. 

Assume that $X\perp ( \{ Y(d,a):( d,a) \in \mathcal{D}\times \mathcal{A}\},D,\{ A( \tilde{d}) :\tilde{d}\in \mathcal{D}\} ) $, $\{ Y(d,a):( d,a) \in \mathcal{D}\times \mathcal{A}\} $ non-stochastic and equal to $\{ \mu (d,a):( {d},a) \in \mathcal{D}\times \mathcal{A}\} $, $A( 0) \sim {\rm Bi}(2,0.3)$, $A( 1) \sim {\rm Bi}(2,0.9)$, $D\sim {\rm Be}( 0.5) $, and $\{ A( 1) ,A( 0),D\} $ are independent random variables. These conditions imply that $A(0) \overset{d}{=}\{A( 0) |D=0\}=\{A|D=0\}\sim {\rm Bi}(2,0.3)$, $A(1) \overset{d}{=}\{A( 1) |D=1\}=\{A|D=1\}\sim {\rm Bi}(2,0.9)$, and $P\{D=1\}=0.5$, as required by the example. Next, we show that the completed example satisfies Assumption \ref{ass:seq-igno}. First, we have that Assumption \ref{ass:seq-igno}(i) holds from the fact that $X$ is independent of the rest of the problem, $\{ Y({d},a):(d,a) \in \mathcal{D}\times \mathcal{A}\} $ is non-stochastic, and $A(d) \perp D$. Second, we have that Assumption \ref{ass:seq-igno}(ii) follows from the fact that $X$ is independent of the rest of the problem and $\{ Y({d},a):(d,a) \in \mathcal{D}\times \mathcal{A}\} $ is non-stochastic.
\end{proof}

\begin{lemma}\label{lem:inter-w-lambda}
Consider the setup in Theorem \ref{thm:inter} and that $A$ is scalar. Then, the coefficients in \eqref{eq:inter-regression} satisfy the following decomposition:
\begin{equation}
\Delta _{\mathrm{inter}}+E[ A] \lambda ~=~\sum_{a\in \mathcal{A}}\omega _{\mathrm{dce}}^{\mathrm{i \star}}(a)(E[Y(1,a)-Y(0,a)])+ \omega _{\mathrm{ind}}^{\mathrm{i \star}}(a) (E[ Y( 0,a) -Y( 0,0) ])~, \label{eq:decomp_phi}
\end{equation}
where
\begin{align}
\Delta &~=~\var(AD)\var(A) -( \cov(DA,A)) ^{2} \notag\\
\Psi &~=~1+\frac{E[ A] }{\Delta }(\cov(A,DA)\cov(A,D)-\var( A) \cov(DA,D)) \notag\\
\omega _{\mathrm{dce}}^{\rm i\star}(a) &~=~ \Psi \omega _{\mathrm{dce}}^{\mathrm{i}}(a)+\frac{E[ A] }{\Delta } \left(\var(A) p\pi _{1}(a)(a-pE[A|D=1]) -\cov(A,DA)(a-E[A])\pi _{1}(a)p \right)\notag\\
\omega _{\mathrm{ind}}^{\rm i\star}(a) &~=~ \Psi \omega _{\mathrm{ind}}^{\mathrm{i}}(a)+\frac{E[ A] }{\Delta }\left( \var(A) p(\pi _{1}(a)a-p_{a}E[A|D=1]) -\cov(A,DA)(a-E[A])p_{a}\right)~. \label{eq:weights_phi}
\end{align}
Moreover, $\sum_{a\in \mathcal{A}}\omega _{\mathrm{dce}}^{\rm i \star}(a)=1$ and $\sum_{a\in \mathcal{A}}\omega _{\mathrm{ind}}^{\rm i \star}(a)=0$. Furthermore, it is possible to have $\omega _{\mathrm{dce}}^{\rm i \star}(a)<0$ and $\omega _{\mathrm{ind}}^{\rm i\star}(a)\neq 0$ for some $a\in \mathcal{A}$.
\end{lemma}
\begin{proof}
By properties of projection,
\begin{equation*}
( \var(W)) ^{-1}( \cov(W,Y)-\cov(W,D)\Delta _{\mathrm{inter} }) ~=~\alpha ~=~(\theta',\lambda')'~.
\end{equation*}
We can use the fact that $A$ is scalar to obtain an explicit formula for $( \var(W)) ^{-1}$. With this expression in hand, we get
\begin{align}
\Delta _{\mathrm{inter}}+E[ A]\lambda
~=~\Psi \Delta _{\mathrm{inter}} +\frac{E[ A] }{\Delta }(\var(A) \cov(DA,Y)-\cov(A,DA)\cov(A,Y))~,\label{eq:phi_mess}
\end{align}
By plugging in the expressions for \eqref{eq:inter_pre1}, \eqref{eq:inter_pre8}, \eqref{eq:inter_pre9} on the right-hand side of \eqref{eq:phi_mess}, imposing Assumption \ref{ass:joint-SOO}, we obtain \eqref{eq:decomp_phi} and \eqref{eq:weights_phi}.

By the definition of $\{(\omega _{\mathrm{dce}}^{\rm i\star}(a),\omega _{\mathrm{ind}}^{\rm i\star}(a)):a \in\mathcal{A}\}$ in \eqref{eq:weights_phi} and repeating arguments 
used in the proof of Theorem \ref{thm:inter_pre}, it is immediate to show that $\sum_{a\in \mathcal{A}}\omega _{\mathrm{dce}}^{\rm i\star}(a)=1$ and $\sum_{a\in \mathcal{A}}\omega _{\mathrm{ind}}^{\rm i\star}(a)=0$.

To conclude, it suffices to find an example in which $\omega _{\mathrm{dce}}^{\rm i\star}(a)<0$ and $\omega _{\mathrm{ind}}^{\rm i\star}(a)\not=0$ for some $a\in \mathcal{A}$. To this end, consider an example with $\{ A|D=0\} \sim {\rm Bi}(2,0.9)$, $\{ A|D=1\} \sim {\rm Bi}(2,0.1)$, and $P\{D=1\} =0.3$, and so $ \mathcal{A}=\{0,1,2\}$. By evaluating this information on \eqref{eq:weights_phi}, we get $\omega_{\rm dce}\approx [-0.2,1.08,0.12]$ and $\omega_{\rm ind}\approx  [ -0.26,0.52,-0.26]$, i.e., (b) and (c) fail.
\end{proof}

\begin{lemma}\label{lem:linear-mu-inter}
Assume the conditions in Theorem \ref{thm:inter}, and that 
\begin{equation}
\mu (d,a)~=~\kappa _{0}+\kappa _{1}d+\kappa _{2}a+\kappa _{3}^{\prime }ad~~\text{ for all }~(d,a)\in \{0,1\}\times \mathcal{A}  \label{eq:imai2}
\end{equation}
for some constants $\kappa _{0},\kappa _{1},\kappa _{2},\kappa _{3}$. Then, the coefficient in \eqref{eq:inter-regression} satisfies $\Delta _{\mathrm{inter}}=\kappa _{1}$, $\theta _{0}=\kappa _{0}$, $\theta =\kappa _{2}$, and $\lambda =\kappa _{3}$. Furthermore, the decomposition in \eqref{eq:inter-decomposition} are $\Delta _{\mathrm{dce}}^{\mathrm{i}}=\kappa _{1}$ and $\Delta _{\mathrm{ind}}^{\mathrm{i}}=0$.
\end{lemma}
\begin{proof}
Assumption \ref{ass:joint-SOO} implies that $E(Y|D=d,A=a)=\mu (d,a)$ which, combined with \eqref{eq:imai2}, implies that the conditional expectation of $Y$ is linear in $(1,a,d,ad)$. From here, the first result follows from the fact that the linear regression consistently estimates the parameters of a linear conditional expectation. The second part follows from combining $\sum_{a\in \mathcal{A}}a\omega _{\mathrm{ind}}^{\mathrm{i}}(a)=\mathbf{0}$ (shown in Theorem \ref{thm:inter_pre}) and \eqref{eq:imai2}.
\end{proof}

\begin{example}\label{ex:basic_with_selection}
Theorem \ref{thm:short} shows that Assumption \ref{ass:joint-SOO} implies that $\Delta_{\rm sel}^{\rm s}=0$. We now provide an example that demonstrates Assumption \ref{ass:basic-SOO} is insufficient to guarantee $\Delta_{\rm sel}^{\rm s} \neq 0$.

Consider the case where $A$ is binary and $D$ is independent of $(A(0), A(1), Y(0,0), Y(1,0), Y(0,1), Y(1,1))$. Assume in addition that $A(0), A(1)$ and $D$ are all independent Bernoulli random variables with $p=0.5$. Note that $A\perp D$ so that $\pi_1(1)=\pi_0(1)$. Next, specify the potential outcomes as $Y(0,0) = Y(1,0) =0$,  $Y(1,1) = 1$, and $Y(0,1) = I\{A(1) > A(0)\}$. 

This example satisfies Assumption \ref{ass:basic-SOO}, as $Y(d) = Y(d, A(d)) \perp D$ for $d = 0,1$. Using \eqref{eq:short-dce-1}–\eqref{eq:short-sel-1} and $\Delta_{\rm short} = E[Y|D=1]- E[Y|D=0]$ yield:
\[
\Delta_{\rm short} = 0.5~, \quad\quad \Delta_{\rm dce}^s = 0.25~, \quad\quad \Delta_{\rm ind}^{\rm s} = 0~, \quad\quad \Delta_{\rm sel}^{\rm s} = 0.25~.
\]
Concretely, while $ \Delta_{\rm ind}^{\rm s} = 0$ follows directly from $\pi_1(1)=\pi_0(1)$, the example illustrates that $\Delta_{\rm sel}^{\rm s} \ne 0$ under Assumption \ref{ass:basic-SOO}. Finally, this example does not contradict the second part of Theorem \ref{thm:short}, since Assumption \ref{ass:joint-SOO} does not hold. 
\end{example}